%% file: main.tex
\documentclass[letterpaper,11pt]{article} %

\usepackage{amsmath}
\usepackage{amsfonts}
\usepackage{amssymb}
\usepackage{tikz}
\usepackage{amsthm}
\usepackage{fullpage}
\usepackage{caption}
\usepackage{bbm}
\usepackage{hyperref, color}
\hypersetup{colorlinks=true,citecolor=red, linkcolor=blue, urlcolor=blue}
\usepackage[linesnumbered,lined,boxed,commentsnumbered,ruled,vlined]{algorithm2e}
\usepackage{enumerate}
\usepackage{bm}
\usepackage{enumitem}
\usepackage{tabularx}
\usepackage{array}
\usepackage[title]{appendix}

\newcommand{\concept}[1]{\emph{{#1}}}

\newcommand{\todo}[1]{\typeout{TODO: \the\inputlineno: #1}\textbf{{\color{red}[[[ #1 ]]]}}}

\newcommand{\MAC}{\textsc{Accept}}
\newcommand{\MRE}{\textsc{Reject}}

\newcommand{\ResolveColoring}{\textsf{Resolve-Coloring}}

\newcommand{\F}{f^v_{c,c'}}
\newcommand{\FA}{f^v_{c^*,c'}}
\newcommand{\FB}{f^v_{c^\circ,c'}}
\newcommand{\FC}{f^v_{c^\star,c'}}
\newcommand{\OPT}{\text{OPT}}
\newcommand{\SOL}{\text{SOL}}

\newcommand{\PAC}{P_{\mathsf{AC}}}
\newcommand{\PRE}{P_{\mathsf{RE}}}

\newcommand{\Resolve}{\textsf{Resolve}}

\newcommand{\EE}[2]{\mathbb{E}_{#1}\left[{#2}\right]}

\newcommand{\HogWild}{\textsl{HogWild!}}

\newcommand{\Y}[2]{{\widehat{Y}_{#1}^{(#2)}}}
\newcommand{\ut}[2]{{t_{#1}^{#2}}}
\newcommand{\ppsl}[2]{{c_{#1}^{#2}}}
\newcommand{\CC}[2]{{\mathcal{C}_{#1}^{#2}}}
\newcommand{\bt}[2]{\beta_{#1}^{#2}}
\newcommand{\DP}[2]{\mathcal{D}_{(#1,#2)}}
\newcommand{\Msg}[2]{\mathsf{Msg}_{#1}^{\rightarrow #2}}
\newcommand{\UD}{\mathsf{UD}}

\newtheorem{theorem}{Theorem}[section]

\newtheorem*{observation*}{Observation}
\newtheorem{claim}[theorem]{Claim}

\newtheorem{lemma}[theorem]{Lemma}
\newtheorem{proposition}[theorem]{Proposition}
\newtheorem{corollary}[theorem]{Corollary}

\theoremstyle{definition}
\newtheorem{definition}{Definition}[section]
\newtheorem{condition}{Condition}
\newtheorem{remark}{Remark}[section]
\newtheorem*{remark*}{Remark}

\title{Distributed Metropolis Sampler with Optimal Parallelism}
\date{}
\author{
Weiming Feng~\thanks{State Key Laboratory for Novel Software Technology, Nanjing University. Emails: {fengwm@smail.nju.edu.cn}, {yinyt@nju.edu.cn}. Supported by the National Key R\&D Program of China 2018YFB1003202 and the National Science Foundation of China under Grant Nos. 61722207 and 61672275.}
\and
Thomas P. Hayes~\thanks{ Department of Computer Science, University of New Mexico. Email:{hayes@cs.unm.edu}. Partially supported by NSF CAREER award CCF-1150281.}
\and
Yitong Yin~\footnotemark[1]
}
  
\begin{document}

\maketitle
\begin{abstract}
The Metropolis-Hastings algorithm is a fundamental Markov chain Monte Carlo  (MCMC) method for sampling and inference. 
With the advent of Big Data, 
distributed and parallel variants of MCMC methods are attracting increased attention.
In this paper, we give a distributed algorithm that can correctly simulate sequential single-site Metropolis chains without any bias in a fully asynchronous message-passing model.
Furthermore, if a natural Lipschitz condition is satisfied by the Metropolis filters, our algorithm can simulate $N$-step Metropolis chains within $O(N/n+\log n)$  rounds of asynchronous communications, {where $n$ is the number of variables}. For sequential single-site dynamics, whose mixing requires $\Omega(n\log n)$ steps, this achieves an optimal linear speedup.
For several well-studied important graphical models, including proper graph coloring, hardcore model, and Ising model, our condition for linear speedup is weaker than the respective uniqueness (mixing) conditions.

The novel idea in our algorithm is to \emph{resolve updates in advance}: the local Metropolis filters can often be executed correctly before the full information about neighboring spins is available. 
This achieves optimal parallelism \mbox{without introducing any bias.}
\end{abstract}


\setcounter{page}{0} \thispagestyle{empty} \vfill
\pagebreak

\tableofcontents{}
\setcounter{page}{0} \thispagestyle{empty} \vfill
\pagebreak

\input{Introduction.tex}

\input{Coloring.tex}

\input{Result.tex}

\input{Analysis.tex}

\input{Appendix.tex}

\bibliographystyle{plain}
\bibliography{refs.bib}
\pagebreak


\end{document}

%% file: Introduction.tex
\section{Introduction}
\label{sec:introduction}
Sampling from joint distributions represented by graphical models is one of the central topics in various fields, including randomized algorithms, statistics, machine learning, and data analysis.
The Metropolis-Hastings method is a fundamental Markov chain Monte Carlo (MCMC) method for sampling. 
Let $\mu$ be a joint distribution for a set $V$ of $n$ random variables, each with domain $[q]$. 
The classic single-site Metropolis chain for sampling from $\mu$ is described as follows.

\vspace{-3pt}
\begin{algorithm}[h]
\SetKwComment{Comment}{$\triangleright$\ }{}
\SetKwInOut{Input}{Input}
\Input{initial configuration $X_0\in[q]^V$}
\For{$t=1$  to $N$}{ 
	pick $v\in V$ uniformly at random\;
	sample a random $c'\in[q]$ and construct $X'\in[q]^V$ by modifying $X_t(v)$ to $c'$\label{alg:metropolis-line-update-1}\;
	with probability $\min\left\{1,\frac{\mu(X')}{\mu(X_t)}\right\}$, set $X_{t+1}\gets X'$; otherwise, set $X_{t+1}\gets X_{t}$\label{alg:metropolis-line-update-2}\;
}
\SetKwInOut{Output}{output}
\caption{single-site Metropolis sampler for $\mu$}\label{Alg:Metropolis}
\end{algorithm}
\vspace{-3pt}

We assume the joint distribution $\mu$ is presented to us as a \concept{graphical model} on an undirected graph $G = (V,E)$.  By this, we mean (see, for instance, \cite{mezard2009information, koller2009probabilistic, feng2019dynamic}) that
the following \concept{conditional independence} property holds for the marginal distribution $\mu_v$ at each node $v\in V$:
\vspace{-5pt}
\[
\vspace{-3pt}
\mu_v\left(\,\cdot\mid X_{V\setminus\{v\}}\right)=\mu_v\left(\,\cdot\mid X_{N(v)}\right).
\]
Here, $N(v) = \{\,u\in V\mid \{u,v\}\in E\,\}$ denotes the neighborhood of $v$ in $G$
and for any subset $S \subset V$, $X_S$ denotes the configuration $X$ restricted to $S$.
In other words, given the spins on its neighbors, the spin at $v$ is conditionally independent
of the spins at all the other vertices in the graph.
Therefore, the Metropolis filter $\min\left\{1,\frac{\mu(X')}{\mu(X_t)}\right\}$ in Algorithm~\ref{Alg:Metropolis} can be computed locally at $v$ by accessing its immediate neighbors as:
$\frac{\mu(X')}{\mu(X_t)}=\frac{\mu_v(X'(v)\mid X_{t}(N_v))}{\mu_v(X_t(v)\mid X_{t}(N_v))}$.

For example, consider \concept{uniform proper graph colorings}, which is a typical graphical model: a $q$-coloring $X\in[q]^V$ is proper for graph $G=(V,E)$ if every pair of adjacent nodes receives two distinct colors.  Here, $\mu$ is the uniform distribution over all proper $q$-colorings of $G$.
Algorithm~\ref{Alg:Metropolis} then instantiates to the following well known Metropolis chain on proper $q$-colorings: at step $t$, it samples a uniform color $c'\in[q]$, and updates the randomly picked vertex $v$'s color to $c'$ as long as $c'\not\in\{X_t(u)\mid u\in N_v\}$.

This classic Metropolis sampler is inherently sequential.
Meanwhile, the  boom in Big Data  applications in contemporary Machine Learning has been drawing increased attention to parallel and distributed algorithms for sampling~\cite{newman2007distributed, doshi2009large, smyth2009asynchronous,yan2009parallel,gonzalez2011parallel,ahmed2012scalable,sa2016ensuring,de2015rapidly,daskalakis2018hogwild,kandasamy2018parallelised}.
In this work, we aim to transform a classic family of  sequential sampling algorithms for graphical models, the single-site Metropolis samplers, to distributed algorithms with ideal parallel speed-up.

It is well known that the sequential Metropolis chain $(X_t)_{t \geq 0}$ can also be obtained from the following continuous-time process $(Y_t)_{t \in \mathbb{R}_{\geq 0}}$, which is parallel in nature:
\begin{itemize}
\item each node $v \in V$ is associated with an i.i.d.~rate-1 Poisson clock; 
\item when the Poisson clock at node $v$ rings, the value at $v$ is updated instantly in the same manner as in the discrete-time Metropolis chain (Line~\ref{alg:metropolis-line-update-1}--\ref{alg:metropolis-line-update-2} in Algorithm~\ref{Alg:Metropolis}). 
\end{itemize}
%
Such continuous-time parallel processes were defined by physicists to study natural evolutions and dynamics \cite{glauber1963time}, even before computer scientists discretized it to a sequential sampling algorithm. 
It is well known that the two processes $(X_t)_{t\ge 0}$ and $(Y_t)_{t \in \mathbb{R}_{\geq 0}}$ are equivalent in the sense that for any $T\ge 0$,  $Y_T$ is identically distributed as $X_{N}$ for $N \sim \mathrm{Pois}(nT)$. 

Although this natural parallel continuous-time process has been known for more than half a century, people do not know how to run it correctly and efficiently in distributed systems.
A major issue arises is that
the updates to variables are non-atomic, 
so that concurrent accesses to {critical regions} consisting of adjacent variables may cause {race conditions} resulting in faulty sampling.
For example, for uniform proper graph coloring, when two adjacent nodes are concurrently updating their own colors based on their knowledges about others' current colors, there is a non-negligible chance that the new coloring becomes improper.

Such inaccuracy in parallel execution of the MCMC sampling can be avoided by a concurrency control that forbids the concurrent updates to adjacent variables; 
see~\cite{gonzalez2011parallel, feng2017sampling}.  However, this results in a 
suboptimal $O(n/\Delta)$ factor of parallel speedup, where $n=|V|$ is the number of variables and $\Delta=\max_{v\in V}|N_v|$ is the maximum degree.
%
Two main approaches have been proposed in order to achieve an optimal linear $O(n)$ speedup:
(1)~the {\HogWild}~method~\cite{smola2010architecture}
assumes stochastic asynchronous schedulers with independently random message delays, and is good for estimating marginal probabilities and Lipschitz functions with small biases~\cite{sa2016ensuring,daskalakis2018hogwild}, although does not always converge to the correct stationary distribution;
(2)~the \textsl{LocalMetropolis} chain~\cite{feng2017sampling, fischer2018simple}, which is a new synchronized parallel Markov chain inspired by the Metropolis chain, 
always converges to the correct stationary distribution, and in particular for sampling uniform proper graph coloring, is rapidly mixing with a $O(n)$-speedup under sufficient conditions for coupling~\cite{fischer2018simple,feng2018distributed}.
Although this second approach gives synchronous algorithms when introduced, by the universal synchronizer for reliable communications~\cite{awerbuch1985complexity}, such synchronous algorithms can be transformed to \concept{fully-asynchronous} algorithms which are correct against adversarial schedulers.
In contrast, the {\HogWild}~method assumes a limited asynchrony that crucially relies on independent random message delays.


\subsection{Our results}
\label{sec:our-result}
We give distributed algorithms for Metropolis sampler that achieve optimal linear parallel speedup under mild conditions.

To state the results in full generality, we present our algorithms as \emph{fully-asynchronous distributed algorithms}, in the message-passing model for distributed computing~\cite{lynch1996distributed, peleg2000distributed, attiya2004distributed},
where the communication network is the graph $G=(V, E)$ for the graphical model,  with nodes as processors and edges as channels, and all communications are asynchronous, with every message delivered within at most one \concept{time unit}.
%
%
%
%
%
For more details of the communication model, see Section~\ref{section-model}.
%


\paragraph{Uniform proper $q$-coloring.}
Let $(X_t^{\mathsf{col}})_{t \geq 0}$ denote the single-site Metropolis chain for uniform proper $q$-coloring on a graph $G=(V, E)$. 
We have the following theorem.
\begin{theorem}\label{thm:main-theorem-color}
Let $n=|V|$ and $\Delta$ the maximum degree of $G$.
Assume that $q \geq \alpha \Delta$ for an arbitrary constant $\alpha > 0$. 
There is a fully-asynchronous distributed algorithm on network $G$ such that:
\begin{itemize}
\item each node $v\in V$ receives an arbitrary $T\ge 0$ as well as its initial color $X^{\mathsf{col}}_0(v)\in[q]$ as inputs;
\item upon termination the algorithm outputs a $X^{\mathsf{col}}_N\in[q]^V$ for some $N\ge nT$, with each node $v\in V$ returning $X^{\mathsf{col}}_N(v)$,  where $X^{\mathsf{col}}_N$ is the $N$-th step  of the single-site Metropolis chain $(X_t^{\mathsf{col}})_{t \geq 0}$;
\item with high probability the algorithm terminates within $O(T+\log n)$ time units, where $O(\cdot)$ hides the constant $\alpha$; and each message contains $O(\log n+\log \lceil T\rceil+\log q)$ bits.
\end{itemize}
\end{theorem}

%

Note that the condition $q=\Omega(\Delta)$ assumed by Theorem~\ref{thm:main-theorem-color} is much weaker than the irreducibility ($q \geq \Delta+2$) as well as the uniqueness ($q \geq \Delta+1$) conditions for proper $q$-colorings.
This is because the goal of the algorithm is to simulate the chain regardless of  mixing.

Given an arbitrarily specified threshold $T\ge 0$, the algorithm returns with high probability a sample $X_N^{\mathsf{col}}$, where $N\ge nT$, within $O(N/n+\log n)$ time units.
In fact, the algorithm achieves this by perfectly simulating the continuous-time chain $(Y_t^{\mathsf{col}})_{t \in \mathbb{R}_{\geq 0}}$ up to time $T'=2T+ 8\ln n$, so that it always returns a $Y_{T'}^{\mathsf{col}}$ upon termination, which with high probability is identical to $X_N^{\mathsf{col}}$ for some $N\ge nT$ due to concentration of $N\sim\mathrm{Pois}(nT')$.
Recall that $\Omega(n\log n)$ is a general lower bound for the mixing time of single-site dynamics~\cite{hayes2007general}. 
Therefore the time bound $(N/n+\log n)$ gives the optimal linear speedup for mixing chains.

\paragraph{General Metropolis samplers.} For general graphical models, we use the following abstract formulation of Metropolis samplers that generalizes Algorithm~\ref{Alg:Metropolis}.
%
%
Let $[q]$ be the finite domain and $G=(V,E)$ the underlying graph of the graphical model, where for each $v\in V$, we use $N_v=N(v)$ to denote the neighborhood of $v$.
%
%
A single-site Metropolis chain with state space $\Omega=[q]^V$ is specified by a sequence $(\nu_v)_{v\in V}$ of \concept{proposal distributions} and a sequence $(f_{c,c'}^v)_{c,c'\in[q],v\in V}$ of \concept{Metropolis filters}, where each $\nu_v$ is a distribution over $[q]$ and each $f_{c,c'}^v:[q]^{N_v}\to[0,1]$ maps local configurations $X_{N_v}\in[q]^{N_v}$ to acceptance probabilities.
The abstract single-site Metropolis chain is as below.

\begin{algorithm}[h]
\SetKwComment{Comment}{$\triangleright$\ }{}
\SetKwInOut{Input}{Input}
\Input{initial configuration $X_0\in[q]^V$}
\For{$t=1$  to $N$}{ 
	pick $v\in V$ uniformly at random and denote $c=X_t(v)$\;
	sample $c'\in[q]$ according to $\nu_v$ and construct $X'\in[q]^V$ by modifying $X_t(v)$ to $c'$\label{alg:abs-metropolis-line-update-1}\;
	with probability $f_{c,c'}^v\left(X_t(N_v)\right)$, set $X_{t+1}\gets X'$; otherwise, set $X_{t+1}\gets X_{t}$\label{alg:abs-metropolis-line-update-2}\;
}
\SetKwInOut{Output}{output}
\caption{single-site Metropolis sampler (abstract version)}\label{Alg:Metropolis-abstract}
\end{algorithm}

The Metropolis chain in Algorithm~\ref{Alg:Metropolis} for sampling from $\mu$ represented by a graphical model on $G=(V,E)$ is a special case of this abstract Metropolis sampler with each $\nu_v$ being the uniform distribution over $[q]$ and $f_{c,c'}^v:[q]^{N_v}\to[0,1]$ defined as $\forall \tau\in[q]^{N_v},\,\, f_{c,c'}^v(\tau) =  \min \left\{ 1, \frac{\mu_v\left(c' \mid \tau\right)}{\mu_v\left(c\mid \tau \right) } \right\}$.\footnote{To have $f_{c,c'}^v$ defined everywhere, we may take the conventions that $\mu_v(\,\cdot\mid\tau)=0$ for the illegal $\tau$ with 0 measure and also $0/0=1$.
Such extension of $f_{c,c'}^v$ to total functions will not affect the definition of the chain over legal states.}
%

We define a Lipschitz condition for the Metropolis filters.

\begin{condition}
\label{condition-Lipschitz}
There is a constant $C > 0$ such that for any $(u,v)\in E$, any $a,b,c\in [q]$,
\begin{align*}
\EE{c' \sim \nu_v}{\delta_{u, a, b}\,\F} \leq \frac{C}{\Delta},	
\end{align*}
where $\Delta=\Delta(G)$ denotes the maximum degree of $G$, and the operator $\delta_{u, a, b}$ is defined as 
\[
\delta_{u, a, b}\,\F\triangleq \max_{(\sigma,\tau)} \left|\F(\sigma) - \F(\tau)\right|
\] 
where 
the maximum is taken over all such $\sigma,\tau\in[q]^{N_v}$ that $\sigma_u=a$, $\tau_u=b$ and $\sigma=\tau$ elsewhere.
\end{condition}


Let $(X_t)_{t \geq 0}$ denote the abstract single-site Metropolis chain in Algorithm~\ref{Alg:Metropolis-abstract} on graph $G=(V,E)$.
We show that this chain can be parallelized with linear speedup as long as  Condition~\ref{condition-Lipschitz} holds.

\begin{theorem}[main theorem]\label{thm:main-theorem}
%
Let $n=|V|$.
Assume that Condition~\ref{condition-Lipschitz} holds. 
There is a fully-asynchronous distributed algorithm on network $G$ such that:
\begin{itemize}
\item each node $v\in V$ receives as inputs an arbitrary threshold $T\ge 0$, an initial value $X_0(v)$, the proposal distribution $\nu_v$ and the Metropolis filters $(f_{c,c'}^v)_{c,c'\in[q]}$ at $v$;
\item upon termination the algorithm outputs $X_N$ for some $N\ge nT$, with each node $v\in V$ returning $X_N(v)$,  where $X_N$ is the $N$-th step  of the single-site chain $(X_t)_{t \geq 0}$ defined in Algorithm~\ref{Alg:Metropolis-abstract};
\item with high probability the algorithm terminates within $O(T+\log n)$ time units,  where $O(\cdot)$ hides the constant $C$ in Condition~\ref{condition-Lipschitz}; and each message contains $O(\log n+\log \lceil T\rceil+\log q)$ bits.
\end{itemize}
\end{theorem}
\noindent
As before, the distributed algorithm in fact simulates the continuous-time chain $(Y_t)_{t\in\mathbb{R}_{\ge 0}}$ up to time $T'=2T+ 8\ln n$, so that the output $Y_{T'}$ is identical to $X_N$ of the discrete-time chain $(X_t)_{t\ge 0}$ for some $N\ge nT$ with high probability.
This is formally proved in Theorem~\ref{thm:main-theorem-Y}.
%

%


\begin{remark}
By the universal synchronizer in~\cite{awerbuch1985complexity}, with reliable communications there is no significant difference between synchronous and asynchronous algorithms in terms of time complexity.
%
Nevertheless, we present our algorithms directly in the asynchronous model without assorting to the synchronizer, because we want to  have the algorithm explicitly specified in a more general model, and the algorithm is in fact simpler to describe asynchronously as event-driven procedure.
\end{remark}

On concrete graphical models,
Condition~\ref{condition-Lipschitz} is easy to apply and gives weak conditions.
In fact, Theorem~\ref{thm:main-theorem-color} is a corollary to Theorem~\ref{thm:main-theorem} on proper $q$-coloring, where Condition~\ref{condition-Lipschitz} translates to $q=\Omega(\Delta)$.
On various other well-studied graphical models, Condition~\ref{condition-Lipschitz} is also weaker than the respective uniqueness conditions.
The followings are two major examples:
\begin{itemize}
\item \textbf{Hardcore model:} 
The distribution $\mu$ is over all configurations in $\{0,1\}^V$ that corresponds to independent sets of $G$. For each configuration $\sigma\in\{0,1\}^V$ that indicate an independent set of $G$, $\mu(\sigma)\propto\lambda^{\sum_{v\in V}\sigma(v)}$, where $\lambda\ge 0$ is the \concept{fugacity}. The natural  Metropolis chain for this model is:
For each $v\in V$, $\nu_v$ is defined over $\{0,1\}$ as $\nu_v(0)=\frac{1}{1+\lambda}$ and $\nu_v(1)=\frac{\lambda}{1+\lambda}$, and
\begin{align*}
\forall c,c'\in\{0,1\}, \tau\in\{0,1\}^{N_v}:\quad 
f_{c,c'}^v(\tau)=\prod_{u\in N_v}I[\tau_u+c'\le 1]. 
\end{align*}
Condition~\ref{condition-Lipschitz} applied on this chain translates to: there is a constant $C>0$ such~that
$\lambda<\frac{C}{\Delta}$, 
while the uniqueness condition is $\lambda<\frac{(\Delta-1)^{\Delta-1}}{(\Delta-2)^{\Delta}}\approx{\frac{\mathrm{e}}{\Delta}}$~\cite{weitz2006counting,sly2010computational}.

\item \textbf{Ising model:}
The distribution $\mu$ is over all configurations in $\{-1,+1\}^V$, such that for each  $\sigma\in\{0,1\}^V$, $\mu(\sigma)\propto\exp\left(\beta\sum_{(u,v)\in E}\sigma_u\sigma_v\right)$, where $\beta\in\mathbb{R}$ is the temperature.
The natural  Metropolis chain for this model is: For each $v\in V$, $\nu_v$ is uniform over $\{-1,1\}$,~and 
\begin{align*}
\forall c,c'\in\{-1,1\}, \tau\in\{-1,1\}^{N_v}:\quad 
f_{c,c'}^v(\tau)=\exp\left(\min\left\{0,\beta (c'-c)\sum_{u\in N_v}\tau_u\right\}\right). %
\end{align*}
Condition~\ref{condition-Lipschitz} applied on this chain translates to: there is a constant $C>0$ such~that
$1-\mathrm{e}^{-2|\beta|}<\frac{C}{\Delta}$,
while the uniqueness condition is  $1-\mathrm{e}^{-2|\beta|}<\frac{2}{\Delta}$~\cite{sinclair2014approximation,sly2014counting}.
\end{itemize}

\section{Model of Asynchrony}\label{section-model}
We adopt the \concept{fully-asynchronous message-passing model} in distributed computing~\cite{lynch1996distributed, peleg2000distributed, attiya2004distributed} as our model of asynchrony. The  network is described by a simple undirected graph $G = (V, E)$, where the nodes represent processors and the edges represent bidirectional channels between them. 

For synchronous systems, communications take place in synchronized \concept{rounds}. 
In each round, each node may send messages to all its neighbors, receive messages from all its neighbors, and perform arbitrary local computation atomically. 
%
%
For the asynchronous systems, there are no synchronous rounds of communications: instead, messages sent from a processor to its neighbor arrive within some finite but unpredictable time, determined by an {adversarial scheduler} who is adaptive to the entire input, satisfying that within a channel all messages sent from a processor will eventually be delivered in the order in which they are sent by the processor (a.k.a.~the {\em reliable FIFO channel}).
Given an execution of an asynchronous algorithm, a \concept{time unit} is defined as the maximum time that elapses between the moment that a message is sent by a processor and the moment that the message is delivered to the receiver. 
The time complexity of an algorithm is measured by the number of time units from the start of the algorithm to its termination in the worst case.

The synchronous model is a special case of the asynchronous model where the scheduler is benign and synchronous.
Any synchronous algorithm can be transformed to a fully-asynchronous algorithm with essentially the same time complexity by the universal synchronizer introduced in~\cite{awerbuch1985complexity}.
The \HogWild~algorithms~~\cite{smola2010architecture, sa2016ensuring,daskalakis2018hogwild} assume a stochastic scheduler where the message delays are independently random.
%
%
 %
%

\section{Technique Overview: \emph{Resolve Updates in Advance}}
\label{sec:algorithm-outline}
We now give an overview of our main technique for distributed sampler: \emph{resolve updates in advance}.


We consider the continuous-time chain $(Y_t)_{t\in\mathbb{R}_{\ge 0}}$, which is the continuous-time version of the abstract Metropolis chain described in Algorithm~\ref{Alg:Metropolis-abstract}. 
Given a network $G=(V,E)$, the chain is specified by the proposal distributions $(\nu_v)_{v\in V}$ and Metropolis filters $(f_{c,c'}^v)_{c,c'\in[q],v\in V}$.
Initially, the chain starts from an arbitrary initial configuration $Y_0\in[q]^V$, and evolves as:
\begin{itemize}
\item each node $v \in V$ is associated with an i.i.d.~rate-1 Poisson clock; 
\item when the Poisson clock at node $v$ rings, 
$Y(v)$ is updated as Line~\ref{alg:abs-metropolis-line-update-1}--\ref{alg:abs-metropolis-line-update-2} of Algorithm~\ref{Alg:Metropolis-abstract}
\end{itemize}
This defines $(Y_t)_{t\in\mathbb{R}_{\ge 0}}$.
For any $t\in\mathbb{R}_{\ge 0}$ and any $v\in V$, $Y_t(v)$ denotes the state of node $v$ at time $t$.

The goal of the distributed algorithm is to \emph{perfectly} simulate the continuous-time chain $(Y_t)_{t\in\mathbb{R}_{\ge 0}}$ in the following sense: initially each node $v\in V$ receives its initial state $Y_0(v)$ and a threshold $T\ge 0$, and when terminates, the algorithm outputs $Y_T$, with each node $v\in V$ outputting $Y_T(v)$.

A straightforward strategy for simulating $(Y_t)_{t \in \mathbb{R}_{\geq 0}}$ up to  time $T$ is to resolve each update at a node $v$ only when $v$ has received the results of all prior adjacent updates from $v$'s neighbors.
Therefore node $v$ knows the neighborhood coloring $Y_t(N_v)$ when $v$ resolves its update at time~$t$.
Enumerating all paths with monotonically decreasing Poisson clocks from time $T$ to time $0$ shows that the algorithm terminates within $O(\Delta T + \log n)$ time units with high probability.
The overhead of max-degree $\Delta$ is necessary in the worst-case because if the scheduler is synchronous, adjacent nodes cannot resolve their updates in the same round.

To achieve an ideal parallel speedup, we must allow adjacent nodes to resolve their own updates simultaneously, yet still guarantee the correctness of the chain.
The algorithm is outlined below.

\begin{center}
\definecolor{light-gray}{gray}{0.95}
\fcolorbox{black}{light-gray}{\shortstack[l]{
\\
\textbf{Algorithm at node $v \in V$:}\\
\parbox[c][][c]{40em}{
\vspace{5pt}
\begin{itemize}[leftmargin=56pt]
\vspace{-7pt}
\item[ \textbf{Phase I:\hspace{5pt}}]\, 
Locally simulate rate-1 Poisson clock for time duration $T$ to generate a sequence of random update times; 
let $m_v$ denote the number of times the clock rings;
and for each update time generate a random proposal from $\nu_v$.
\vspace{3pt}\\
\vspace{3pt}Send the initial value, update times and proposals to all neighbors.\\
Upon receiving above informations from all neighbors, enter \textbf{Phase II}.
\vspace{-8pt}
\item[ \textbf{ Phase II:}] \,
{For} $i=1$ to $m_v$ {do}:
\vspace{-6pt}
\begin{itemize}[leftmargin=40pt]
\item[] Keep listening to the channels from all neighbors. 
\vspace{-3pt}
\item[($\star$)] As soon as $v$ gets \textbf{\emph{enough information}} to resolve its $i$-th update: 
\begin{list}{}{}
\vspace{-3pt}
\item Resolve the $i$-th update of $v$ and send the update decision \\
(``\MAC'' or ``\MRE'') to all neighbors. 
\end{list}
\end{itemize}
\end{itemize}
\vspace{-7pt}
}
}
}
\end{center}

For convenience of exposition, the algorithm is divided into two \concept{phases}.
In the first phase, the algorithm at each node $v\in V$ locally generates its own random choices and shares them with all its neighbors.
%
%
Therefore in the beginning of the second phase, node $v$ knows the initial states, all update times and proposals before time $T$ of itself and all its neighbors.
%

In the second phase, the algorithm at each node $v\in V$ resolves the updates at $v$ one by one.
The crucial step is Line~$(\star)$, when node $v$ resolves an update. 
%
For instance, considering the chain for uniform proper $q$-coloring,  for an update at time $t$ in the chain with $c'\in[q]$ as the proposed color, such an update is resolved once one of the following two events occurs:
\begin{itemize}
\item $\forall u \in N_v$, $Y_{t}(u) \neq c'$, in which case the proposed color $c'$ must be accepted;
\item $\exists u \in N_v$, $Y_{t}(u) = c'$, in which case the proposed color $c'$ must be rejected.
\end{itemize}
Our key observation is: the update can be resolved \emph{in advance}, before the neighborhood coloring $Y_t(N_v)$ fully known to $v$.
This is because node $v$ can infer all the colors that any neighbor $u\in N_v$ can possibly be at time $t$ because $v$ knows $u$'s all update times and proposed colors.
The above two events can be checked with such information by chance before $Y_t(N_v)$ is completely known.

For general continuous-time Metropolis chain $(Y_t)_{t\in\mathbb{R}_{\ge 0}}$,
to resolve an update at time $t$ in the chain with $c'\in[q]$ as the proposal, the algorithm at node $v$ needs to flip a coin with bias $f_{c,c'}^v(Y_t(N_v))$, where $c=Y_{t-\epsilon}(v)$,  to determine the acceptance/rejection of the update. 
As before, node $v$ can infer all possible values of $Y_t(u)$ for each neighbor $u\in N_v$, which composes a (product) space of possible local configurations $Y_t(N_v)$, each corresponding to a possible biased coin.
%
We can couple all these possible biased coins together, so that the coin flipping can be carried out by chance before its true bias is fully known, yet still being correct in any case. 
Such coupled coin flipping can even be implemented computationally efficiently because the space of uncertainty is product.
The algorithm is formally described in Section~\ref{sec:distributed-simulation}.

The correctness of the algorithm can be verified by a perfect coupling between the algorithm and the continuous-time chain $(Y_t)_{t\in\mathbb{R}_{\ge 0}}$.
Its efficiency, is established by a \concept{dependency chain} argument.
A dependency chain is a sequence of adjacent updates, 
in which the resolution of each update is \emph{triggered} by the resolution of the previous update.
The time complexity of the algorithm is bounded by the length of longest dependency chain because each dependency takes place within at most one time unit. 
We show that when Condition~\ref{condition-Lipschitz} holds, the probability of a dependency chain decays exponentially as its length grows.
The detailed analyses are in Section~\ref{sec:correctness} and Section~\ref{sec:running-time}.

%% file: Coloring.tex
\section{Distributed Simulation of Continuous-Time Chains}
\label{sec:distributed-simulation}


We now formally describe the \textbf{main algorithm} outlined in last section that simulates the continuous-time chain $(Y_t)_{t\in\mathbb{R}_{\ge0}}$ up to time $T$ in the asynchronous communication model.
Fix any node $v\in V$.
The algorithm at node $v$ consists of two phases.
In \textbf{Phase I}, node $v$ first locally generates the following random values:
\begin{itemize}
\item the random update times $0<\ut{v}{1}<\cdots<\ut{v}{m_v}<T$, generated independently by a rate-1 Poisson clock, where $m_v\sim\mathrm{Pois}(T)$ denotes the number of times the clock rings before time~$T$. 
\item the random proposals $\ppsl{v}{1},\ppsl{v}{2},\ldots,\ppsl{v}{m_v}\in[q]$ sampled i.i.d.~according to distribution $\nu_v$. 
\end{itemize}
Then node $v$  sends these random values and its initial value $Y_0(v)$ to all its neighbors.
The update times $\ut{v}{i}$ can be represented with bounded precision as long as being  distinct and preserving relative order among adjacent nodes, which is enough for correctly simulating the chain $(Y_t)_{t \in[0,T]}$.
This can be done within $O(T+\log n)$ time unites with high probability, with each message of size $O(\log n+\log \lceil T\rceil+\log q)$.\footnote{The update times $\ut{v}{i}$ are truncated adaptively with bounded precision when sent to neighbors, ensuring that among neighbors the truncated update times are distinct and having the same relative order  as before truncation.}
%
%
%
Once node $v$ receives all such informations from all its neighbors, it enters \textbf{Phase II}.

With the random update times and proposals $\left(\ut{v}{i},\ppsl{v}{i}\right)$, 
the continuous-time chain $(Y_t)_{t\in[0,T]}$ can be generated 
starting from the initial configuration $Y_0\in[q]^V$,  
where at each update time $t=\ut{v}{i}$ for some $v\in V$ and $1\le i\le m_v$:
\begin{align}
Y_{t}(v)=
\begin{cases}
\ppsl{v}{i} & \text{with prob.~}f_{c,c'}^{v}(Y_{t}(N_v))\\
Y_{t-\epsilon}(v) & \text{with prob.~}1-f_{c,c'}^{v}(Y_{t}(N_v))
\end{cases},
\quad\text{where }c'=\ppsl{v}{i}\text{ and }c=Y_{t-\epsilon}(v),
\label{eq:eq-def-chain-Y}
\end{align}
and in general for any $0<t\le T$ and $v\in V$: $Y_t(v)=Y_{\ut{v}{i}}(v)$ for the $i$ satisfying $\ut{v}{i}\le t< \ut{v}{i+1}$.
Furthermore, for any $v\in V$ and $0\le i\le m_v$ we denote the state of $v$ right after its $i$-th update as:
\begin{align}
Y_v^{(i)}\triangleq Y_{\ut{v}{i}}(v).\label{eq:def-step-Y}
\end{align}

%
The \textbf{Phase II} of the algorithm at node $v$ is described  in Algorithm~\ref{ResolveAllUpdates}.

\begin{algorithm}[h]
\SetKwComment{Comment}{\quad$\triangleright$\ }{}
\setcounter{AlgoLine}{0}
\SetKwInOut{Initialization}{Assumption}
\Initialization{node $v\in V$ knows the initial values $Y_0(u)$ and the lists of update times and proposals $\left(\ut{u}{i},\ppsl{u}{i}\right)_{1\le i\le m_u}$ of all $u\in N_v\cup\{v\}$.}
Initialize $\widehat{Y}_v^{(0)}$,   
$\mathbf{j}=(j_u)_{u\in N_v}$ and $\mathbf{Y}=\left(\Y{u}{j}\right)_{u\in N_v, 0\le j\le j_{u}-1}$ respectively as:\\
\quad\,\, $\widehat{Y}_v^{(0)}\gets Y_0(v)$
and for all $u\in N_v$: $j_u\gets {1}$ and $\Y{u}{0}\gets Y_0(u) $;\\ 

\For{$i=1$  to $m_v$}{ 
	$\left(\Y{v}{i}, \mathbf{j}, \mathbf{Y}\right) \gets \Resolve\left(i, \Y{v}{i-1}, \mathbf{j}, \mathbf{Y}\right)$\Comment{resolve the $i$-th update}
}
\SetKwInOut{Output}{output}
\Return{$\Y{v}{m_v}$\;}
\caption{Pseudocode for \textbf{Phase II} of the main algorithm at node $v$}\label{ResolveAllUpdates}
\end{algorithm}

\begin{remark}[\textbf{synchronization between phases}]
Different nodes may enter Phase II at different time due to asynchrony. By definition, no message is ever sent from a Phase-I node to a Phase-II node. The messages sent from a Phase-II node $u$ to a Phase-I node $v$ are stored in a queue at $v$ and processed by $v$ in the same order they are received once $v$ enters its own Phase II.
\end{remark}

%
%
The goal of the algorithm at node $v$ is to output the value $Y_{v}^{(m_v)}=Y_T(v)$.
It does this by resolving all $m_v$ updates of $v$ one by one in the order they are issued.
In the $i$-th iteration, node $v$ has resolved the first $i - 1$ updates and maintains its current state as $\Y{v}{i-1}$; and within the iteration, it resolves the $i$-th update and updates its current state $\Y{v}{i-1}$ to a new  $\Y{v}{i}$.
This naturally defines a continuous-time chain $(\widehat{Y}_t)_{t\in[0,T]}$: 
\begin{align}
\forall 0\le t\le T\text{ and }\forall v\in V:\quad \widehat{Y}_t(v)=\Y{v}{i}\text{ where $i$ satisfies }\ut{v}{i}\le t< \ut{v}{i+1}.\label{eq-def-chain-hat-Y}
\end{align}

In Algorithm~\ref{ResolveAllUpdates}, the current state $\Y{v}{i}$ of node $v$ is updated with the assistance of the following auxiliary data structures maintained at node $v$:
\begin{itemize}
\item a tuple $\mathbf{j} = (j_u)_{u \in N_v}$, such that each $j_u$ denotes that from $v$'s perspective, 
the neighbor $u$ is resolving its $j_u$-th update, and the outcomes of its first $(j_u-1)$ updates are known to $v$;
\item  a tuple $\mathbf{Y} = \left(\Y{u}{j}\right)_{u \in N_v, 0\leq j \leq j_u - 1}$, such that each $\Y{u}{j}$ memorizes for $v$ the correct state of $u$ in the algorithm right after resolving its $j$-th update.
\end{itemize}
Given any such $\mathbf{j} = (j_u)_{u \in N_v}$ and $\mathbf{Y} = \left(\Y{u}{j}\right)_{u \in N_v, 0\leq j \leq j_u - 1}$,  for any neighbor $u\in N_v$ and any time $t<\ut{u}{j_u}$, due to~\eqref{eq-def-chain-hat-Y}, node $v$ can infer the precise value of $\widehat{Y}_t(u)$ as: $\widehat{Y}_t(u) = \Y{u}{k-1}$ where $k$ satisfies $\ut{u}{k -1} \leq t <\ut{u}{k}$.

For the time $t\ge\ut{u}{j_u}$, node $v$ can no longer always infer the exact value of $\widehat{Y}_t(u)$ for a neighbor $u\in N_v$ based on the partial information encoded by $\mathbf{j} = (j_u)_{u \in N_v}$ and $\mathbf{Y} = \left(\Y{u}{j}\right)_{u \in N_v, 0\leq j \leq j_u - 1}$. 
However, it can narrow down the possible values of $\widehat{Y}_t(u)$ to a subset of $[q]$.

\begin{definition}[\textbf{set of possible states}]
\label{definition-set-of-possible-values}
Fixed a node $v \in V$ and given its current $\mathbf{j} = (j_u)_{u \in N(v)}$ and $\mathbf{Y} = \left(\Y{u}{j}\right)_{u \in N_v, 0\leq j \leq j_u - 1}$,
for any neighbor $u \in N_v$ and any time $0\leq t < T$,  the set of possible states	$\mathcal{S}_{t}(u) \subseteq [q]$ for node $u$ at time $t$ with respect to $\mathbf{j}$ and $\mathbf{Y}$ is defined as 
\begin{align}
\label{eq-def-Stu}
\mathcal{S}_{t}(u) \triangleq \begin{cases}
\vspace{3pt}\left\{\,\Y{u}{k-1}\,\right\} 
{ \text{ where }\ut{u}{k -1} \leq t <\ut{u}{k}}
&\text{if } 0 \leq t < \ut{u}{j_u}	\\
\left\{\Y{u}{j_u-1}\right\} \cup \left\{ \ppsl{u}{k} \mid \ut{u}{j_u} \leq \ut{u}{k} \leq t \right\}  &\text{if } \ut{u}{j_u} \leq t <T.
\end{cases}
\end{align}
\end{definition}
For Metropolis chains, it holds for each transition that $\Y{u}{j}\in\left\{\ppsl{u}{j},\Y{u}{j-1}\right\}$, i.e.~the state of $u$ is either changed to the proposal or unchanged.
It is then easy to verify that $\widehat{Y}_t(u) \in \mathcal{S}_t(u)$ always holds for any $u \in N_v$ any $0 \leq t < T$.
This guarantees the soundness of the definition.
Furthermore, node $v$ can locally compute sets $\mathcal{S}_t(u)$ for all $u \in N_v$ and  $0\leq t <T$ based on its current $\mathbf{j}$ and $\mathbf{Y}$.

%
Next, we show how to use the partial information of $Y_t(u)$ captured by $\mathcal{S}_t(u)$ to resolve an update with filter $f_{c,c'}^v(Y_t(N_v))$ before $Y_t(N_v)$ is fully known.


\subsection{Example: proper graph coloring}
\label{section-coloring}

We first give the subroutine $\Resolve$ on a special case: the uniform proper $q$-coloring of graph $G=(V,E)$. 
To be distinguished from the general case, we use $\ResolveColoring(i, \widehat{Y}(v), \mathbf{j}, \mathbf{Y})$ to denote the subroutine for this special case, 
described in Algorithm~\ref{ResolveColoring}.

\begin{algorithm}[h]
\SetKwInOut{Input}{input}
\SetKwIF{upon}{}{}{upon}{do}{}{}{}
\Input{index $i$ of the current update; current color $\widehat{Y}(v)$ of $v$; the tuples $\mathbf{j} = (j_u)_{u \in N_v}$ and $\mathbf{Y} = \left(\Y{u}{j}\right)_{u \in N_v, 0\leq j \leq j_u - 1}$ of the current steps and historical colors of $v$'s neighbors;}
construct $\mathcal{S}(u)=\mathcal{S}_{\ut{v}{i}}(u)$ as~\eqref{eq-def-Stu} for all $u\in N_v$\label{alg:coloring-line-1}\;
\upon{$\ppsl{v}{i} \notin \bigcup_{u \in N_v}\mathcal{S}(u)$\label{alg:coloring-condition-1} }{
		send message ``$\MAC$'' to all neighbors $u\in N_v$\;
	    \Return{ $\left(\ppsl{v}{i}, \mathbf{j}, \mathbf{Y}\right)$ }\;
}
\upon{$\exists u \in N_v$ s.t. $\mathcal{S}(u) = \left\{\,\ppsl{v}{i}\,\right\}$\label{alg:coloring-condition-2}}{
		send message ``$\MRE$'' to all neighbors $u\in N_v$\;
		\Return{ $\left(\widehat{Y}(v), \mathbf{j}, \mathbf{Y}\right)$ }\;
}
\upon{receiving ``\MAC'' from a neighbor $u \in  N_v$}{
	 modify $\mathbf{Y}$ by setting $\Y{u}{j_u}=\ppsl{u}{j_u}$\;
	$j_u \gets j_u + 1$\;
	recompute $\mathcal{S}(u)=\mathcal{S}_{\ut{v}{i}}(u)$ as~\eqref{eq-def-Stu} \label{alg:coloring-update-1} w.r.t.~new $\mathbf{j}$ and $\mathbf{Y}$\;
}
\upon{receiving ``\MRE'' from a neighbor $u \in  N_v$}{
	 modify $\mathbf{Y}$ by setting $\Y{u}{j_u}=\Y{u}{j_u-1}$\;
	$j_u \gets j_u + 1$\;
	recompute $\mathcal{S}(u)=\mathcal{S}_{\ut{v}{i}}(u)$ as~\eqref{eq-def-Stu} \label{alg:coloring-update-2} w.r.t.~new $\mathbf{j}$ and $\mathbf{Y}$\;
}
\caption{$\ResolveColoring(i,\widehat{Y}(v), \mathbf{j}, \mathbf{Y})$ at node $v$}\label{ResolveColoring}
\end{algorithm}

Algorithm~\ref{ResolveColoring} resolves the $i$-th update $\left(\ut{v}{i},\ppsl{v}{i}\right)$ for proper coloring.
Node $v$ maintains for each neighbor $u \in N_v$ a set $\mathcal{S}(u)=\mathcal{S}_{\ut{v}{i}}(u)$ of possible colors of $u$ at time $\ut{v}{i}$ based on the $\mathbf{j}$ and $\mathbf{Y}$ maintained by $v$. 
The algorithm is event-driven:
if more than one events occur simultaneously, they are responded in the order listed in the algorithm.
%
%
Node $v$ is constantly monitoring the sets $\mathcal{S}(u)$ for all $u\in N_v$. 
The update $\left(\ut{v}{i},\ppsl{v}{i}\right)$ is resolved once one of the following two events occurs:
%

\begin{itemize}
\item $\ppsl{v}{i} \notin \bigcup_{u \in N_v}\mathcal{S}(u)$, in which case the proposed color $\ppsl{v}{i}$ is determined to not conflict with all possible colors of any neighbors at time $\ut{v}{i}$, thus $\ppsl{v}{i}$ must be accepted;
\item $\exists u \in N_v$ s.t. $\mathcal{S}(u) = \left\{\ppsl{v}{i}\right\}$, in which case the proposed color $\ppsl{v}{i}$ is determined to be blocked by neighbor $u$'s color at time $\ut{v}{i}$, thus $\ppsl{v}{i}$ must be rejected.
\end{itemize}
These two events are mutually exclusive. 
Also the set $\mathcal{S}(u)$ of possible colors for each neighbor $u\in N_v$ is updated correctly once a message ``\MAC{}'' or ``\MRE{}'' is received  from $u$. 
Eventually, at least one of the above two events must occur so that the algorithm must terminate. 
%

%% file: Result.tex
\newcommand{\GeneralResult}{
\section{General Results (Revised)}
\label{section-main-result}
\todo{A paragraph say: the technique in graph coloring can be applied to general Metropolis chains. We have the following result about simulating the general Metropolis chains. The general algorithm is given in the next section.}

We define a Lipschitz condition for the Metropolis filters.
Let $f:[q]^d\to \mathbb{R}$ be a $d$-variate total function. Given any $1\le i\le d$,  $a,b\in[q]$, we define the operator $\delta_{i, a, b}$ on function $f$ as:
\begin{align*}
\delta_{i, a, b}\, f \triangleq \max_{\mathbf{x},\mathbf{y}} |f(\mathbf{x}) - f(\mathbf{y})|,
\end{align*}
where the  maximum is taken over all pairs of $(\mathbf{x},\mathbf{y})$ satisfying $x_i=a$, $y_i=b$ and $x_j=y_j$ for all $j\neq i$.

Consider a (discrete-time or continuous-time) Metropolis chain with proposal distributions $(\nu_v)_{v\in V}$ and Metropolis filters $(f_{c,c'}^v)_{v\in V,\, c,c'\in[q]}$. Recall that each Metropolis filter $f_{c,c'}^v$ is a function $f_{c,c'}^v: [q]^{N_v} \to [0,1]$.
We define the following Lipschitz-style condition.
\begin{condition}
\label{condition-Lipschitz}
There is a constant $C > 0$ such that for any $(u,v)\in E$, any $a,b,c\in [q]$, it holds that
\begin{align*}
\EE{c' \sim \nu_v}{\delta_{u, a, b}\,\F} \leq \frac{C}{\Delta},	
\end{align*}
where $\Delta=\Delta(G)$ denotes the maximum degree of $G$.
\end{condition}


We show that any continuous-time Metropolis chain satisfying Condition~\ref{condition-Lipschitz} can be simulated efficiently by a distributed algorithm, even against adversarial asynchronous schedulers.

%


\begin{theorem}[general theorem]\label{thm:general-theorem}
Let $(Y_t)_{t \in\mathbb{R}_{\ge 0}}$ be  a continuous-time Metropolis chain.
There is a fully-asynchronous distributed algorithm that perfectly simulates $(Y_t)_{t \in\mathbb{R}_{\ge 0}}$. Moreover, if Condition~\ref{condition-Lipschitz} holds for $(Y_t)_{t \in\mathbb{R}_{\ge 0}}$, then given any $T\ge 0$, with high probability  the algorithm terminates and output $Y_T$ within $O(T + \log n)$ time units, with each message of size $O(\log n+\log q)$.
\end{theorem}

The notions of asynchronous  distributed algorithm and distributed simulation of continuous-time Metropolis chains used in above theorem are as defined in Section~\ref{section-problem}. 
The $O(\cdot)$ hides (linearly) the constant $C$ in Condition~\ref{condition-Lipschitz}.

Due to the well-known relation between discrete-time and continuous-time Metropolis chains, we have the following result on simulating single-site Metropolis chain.

\begin{corollary}\label{thm:general-corollary}
Assume that Condition~\ref{condition-Lipschitz} holds for the single-site Metropolis chain $(X_t)_{t\geq 0}$.
There is a fully-asynchronous distributed algorithm that simulates $(X_t)_{t\geq 0}$ such that given any threshold $M \in \mathbb{N}$, with high probability the algorithm outputs $X_N \in [q]^V$ for some $N \geq M$ within $O(M/n + \log n)$ time units, with each message of size $O(\log n+\log q)$.
\end{corollary}

To simulate $(X_t)_{t \geq 0}$ for at least $M$ steps, we use the algorithm in Theorem~\ref{thm:general-theorem} to simulate its continuous version $(Y_t)_{t\in \mathbb{R}_{\geq 0}}$ with $T = M / n + O(\log n)$. 
The corollary holds due to the relation in~\eqref{eq-con-disc-chain} and the concentration property of Poisson distribution. 

Hence, our algorithm simulates $\Omega(n\log n)$-step sequential Metropolis chains distributedly with linear speedup.
On various well-studied graphical models, this gives efficient distributed simulations of respective Metropolis chains under following \mbox{conditions}:
\begin{itemize}
\item \textbf{Proper $q$-coloring:} The natural Metropolis chain for proper $q$-coloring is as follows.
For each $v\in V$, $\nu_v$ is the uniform over $[q]$ and 
\begin{align*}
\forall c,c'\in[q], \tau\in[q]^{N_v}:\quad f_{c,c'}^v(\tau)=\prod_{u\in N_v}{I}[\tau_u\neq c'].
\end{align*}
Over proper $q$-colorings, this chain behaves identically as the chain described in~\eqref{eq-def-accept-prob}.
Condition~\ref{condition-Lipschitz} applied on this chain translates to: there is a constant $\alpha>0$ such~that
\begin{align*}
q\ge\alpha\Delta. 
\end{align*}
In contrast, the uniqueness condition is $q\ge \Delta+1$~\cite{jonasson2002uniqueness,galanis2015inapproximability}.

\item \textbf{Hardcore model:} 
The distribution $\mu$ is over all configurations in $\{0,1\}^V$ that corresponds to independent sets of $G$. For each configuration $\sigma\in\{0,1\}^V$ that indicate an independent set of $G$, $\mu(\sigma)\propto\lambda^{\sum_{v\in V}\sigma(v)}$, where $\lambda\ge 0$ is the \concept{fugacity}. The natural  Metropolis chain for this model is:
For each $v\in V$, $\nu_v$ is a distribution over $\{0,1\}$ with $\nu_v(0)=\frac{1}{1+\lambda}$ and $\nu_v(1)=\frac{\lambda}{1+\lambda}$, and
\begin{align*}
\forall c,c'\in\{0,1\}, \tau\in\{0,1\}^{N_v}:\quad 
f_{c,c'}^v(\tau)=\prod_{u\in N_v}I[\tau_u+c'\le 1]. 
\end{align*}
Condition~\ref{condition-Lipschitz} applied on this chain translates to: there is a constant $C>0$ such~that
\[
\lambda<\frac{C}{\Delta}, 
\] 
while the uniqueness condition is $\lambda<\frac{(\Delta-1)^{\Delta-1}}{(\Delta-2)^{\Delta}}\approx{\frac{\mathrm{e}}{\Delta}}$~\cite{weitz2006counting,sly2010computational}.

\item \textbf{Ising model:}
The distribution $\mu$ is over all configurations in $\{-1,+1\}^V$, such that for each  $\sigma\in\{0,1\}^V$, $\mu(\sigma)\propto\exp\left(\beta\sum_{(u,v)\in E}\sigma_u\sigma_v\right)$, where $\beta\in\mathbb{R}$ is the temperature.
The natural  Metropolis chain for this model is: For each $v\in V$, $\nu_v$ is uniform over $\{-1,1\}$,~and 
\begin{align*}
\forall c,c'\in\{-1,1\}, \tau\in\{-1,1\}^{N_v}:\quad 
f_{c,c'}^v(\tau)=\exp\left(\min\left\{0,\beta (c'-c)\sum_{u\in N_v}\tau_u\right\}\right). %
\end{align*}
Condition~\ref{condition-Lipschitz} applied on this chain translates to: there is a constant $C>0$ such~that
\[
1-\mathrm{e}^{-2|\beta|}<\frac{C}{\Delta},
\]
while the uniqueness condition is  $1-\mathrm{e}^{-2|\beta|}<\frac{2}{\Delta}$~\cite{sinclair2014approximation,sly2014counting}.
\end{itemize}
For above graphical models, Condition~\ref{condition-Lipschitz} is much weaker than the uniqueness (mixing) conditions, because our goal is to simulate the sequential chain regardless of its mixing. 
}

\subsection{General Metropolis chain}
\label{section-general}

We then give the subroutine $\Resolve (i, \widehat{Y}(v), \mathbf{j}, \mathbf{Y})$ for resolving the $i$-th update $\left(\ut{v}{i},\ppsl{v}{i}\right)$ at node $v$ in a general Metropolis chain. 
Throughout, we denote the current and proposed states of $v$ respectively as:
\[
c\triangleq \widehat{Y}(v)\text{ and }c'\triangleq \ppsl{v}{i}.
\]
The algorithm then simulates a coin flipping with bias $f_{c,c'}^v\left(\widehat{Y}_\ut{v}{i}(N_v)\right)$, to determine whether the proposal $\ppsl{v}{i}$ is accepted,
%
%
while the configuration $\widehat{Y}_\ut{v}{i}(N_v)$ is only partially known to node $v$.

In particular, given the current $\mathbf{j}$ and $\mathbf{Y}$ (as inputs to the subroutine $\Resolve$), the set $\mathcal{S}_{\ut{v}{i}}(u)$ of possible states of each neighbor $u\in N_v$ at the time $\ut{v}{i}$, can be constructed as Definition~\ref{definition-set-of-possible-values}.
We further define the set of all possible configurations on the neighborhood $N_v$ as
\begin{align}
\label{eq-def-CC}
\CC{v}{i} \triangleq \bigotimes_{u \in N_v}\mathcal{S}_{\ut{v}{i}}(u).
\end{align}
Note that $\CC{v}{i}$ must contain the correct configuration $\widehat{Y}_{\ut{v}{i}}(N_v)$ because $\widehat{Y}_{\ut{v}{i}}(u)\in \mathcal{S}_{\ut{v}{i}}(u)$ for every $u\in N_v$ due to the soundness of Definition~\ref{definition-set-of-possible-values}.
Besides, $\CC{v}{i}$ may contain various other candidate configurations $\tau\in[q]^V$, each corresponding to a coin with bias $f_{c,c'}^v(\tau)$.
The subroutine resolves the update in advance and correctly by coupling all these coins.

Specifically, for each $\tau \in \CC{v}{i}$, we define the indicator random variable $I_{\mathsf{AC}}(\tau) \in \{0, 1\}$   as 
$\Pr[\,I_{\mathsf{AC}}(\tau) = 1\,] = f^v_{c,c'}(\tau)$. 
The coins $I_{\mathsf{AC}}(\tau)$ for all $\tau \in \CC{v}{i}$ are coupled as follows: 
\begin{align*}
I_{\mathsf{AC}}(\tau) = \begin{cases}
 1 &\text{if } \beta<f^v_{c,c'}(\tau)\\
 0 &\text{if } \beta \geq f^v_{c,c'}(\tau)
 \end{cases},
 \quad\,\beta \text{ is uniformly distributed over } [0,1).
\end{align*}
Since the true $\widehat{Y}_{\ut{v}{i}}(N_v) \in \CC{v}{i}$, the update  can be resolved once all  the indicator random variables $I_{\mathsf{AC}}(\tau)$ for $\tau \in \mathcal{C}_{t(v,i)}$ are perfectly coupled, i.e.~$\forall \tau_1, \tau_2 \in \CC{v}{i}:I_{\mathsf{AC}}(\tau_1) = I_{\mathsf{AC}}(\tau_2) $.

To check whether all the indicator random variables are perfectly coupled,
we define the minimum acceptance probability $\PAC$ and  the minimum rejection probability $\PRE$ as:
\begin{equation}
\label{eq-definition-pac-pre}
\begin{split}
\PAC &\triangleq \min_{\tau \in \CC{v}{i}}f^v_{c,c'}(\tau);\\
\PRE &\triangleq \min_{\tau \in \CC{v}{i}}\left(1-f^v_{c,c'}(\tau)\right) = 1 - \max_{\tau \in \CC{v}{i}}f^v_{c,c'}(\tau).
\end{split}
\end{equation}
Then all coins $I_{\mathsf{AC}}(\tau)$ for $\tau \in \CC{v}{i}$ are perfectly coupled if $\beta < \PAC$ or $\beta \geq 1 - \PRE$. The former case corresponds to the event that all $I_{\mathsf{AC}}(\tau) = 1$, while the latter corresponds to that all $I_{\mathsf{AC}}(\tau) = 0$.


\begin{algorithm}[h]
\SetKwInOut{Input}{input}
\SetKwIF{upon}{}{}{upon}{do}{}{}{}
\Input{index $i$ of the current update; current value $\widehat{Y}(v)$ of $v$; the tuples $\mathbf{j} = (j_u)_{u \in N_v}$ and $\mathbf{Y} = \left(\Y{u}{j}\right)_{u \in N_v, 0\leq j \leq j_u - 1}$ of the current steps and historical values of $v$'s neighbors;}
sample $\beta \in [0,1)$ uniformly at random\label{alg:general-sample}\; 
compute $\PAC$ and $\PRE$ as~\eqref{eq-definition-pac-pre}\label{alg:general-pac-pre}, where $\mathcal{S}_{\ut{v}{i}}(u)$ for each $u\in N_v$ is constructed as~\eqref{eq-def-Stu}\;
\upon{$\beta < \PAC$\label{alg:general-condition-1} }{
		send message ``$\MAC$'' to all neighbors $u\in N_v$\;
	    \Return{ $\left(\ppsl{v}{i}, \mathbf{j}, \mathbf{Y}\right)$ }\;
}
\upon{$\beta \geq 1 - \PRE$\label{alg:general-condition-2}}{
		send message ``$\MRE$'' to all neighbors $u\in N_v$\;
		\Return{ $\left(\widehat{Y}(v), \mathbf{j}, \mathbf{Y}\right)$ }\;
}
\upon{receiving ``\MAC'' from a neighbor $u \in  N_v$\label{alg:general-receive-ac}}{
	 modify $\mathbf{Y}$ by setting $\Y{u}{j_u} = \ppsl{u}{j_u}$\;
	$j_u \gets j_u + 1$\;
	recompute $\PAC$ and $\PRE$ as~\eqref{eq-definition-pac-pre}\label{alg:general-pac-pre-1} with $\mathcal{S}_{\ut{v}{i}}(u)$ reconstructed as~\eqref{eq-def-Stu}\label{alg:general-update-s-1} w.r.t.~new $\mathbf{j}$ and $\mathbf{Y}$\;
}
\upon{receiving ``\MRE'' from a neighbor $u \in  N_v$\label{alg:general-receive-re}}{
	 modify $\mathbf{Y}$ by setting $\Y{u}{j_u} = \Y{u}{j_u-1}$\;
	$j_u \gets j_u + 1$\;
	recompute $\PAC$ and $\PRE$ as~\eqref{eq-definition-pac-pre}\label{alg:general-pac-pre-2} with $\mathcal{S}_{\ut{v}{i}}(u)$ reconstructed as~\eqref{eq-def-Stu}\label{alg:general-update-s-2} w.r.t.~new $\mathbf{j}$ and $\mathbf{Y}$\;
}
\caption{$\Resolve(i,\widehat{Y}(v), \mathbf{j}, \mathbf{Y})$ at node $v$}\label{Resolve}
\end{algorithm}

The procedure for $\Resolve(i, \widehat{Y}(v), \mathbf{j}, \mathbf{Y})$ at node $v$ is described in Algorithm~\ref{Resolve}.
In the algorithm, a random number $\beta\in[0,1)$ is sampled only once (hence the coupling) in the beginning and used during the entire execution of the subroutine $\Resolve(i, \widehat{Y}(v), \mathbf{j}, \mathbf{Y})$ for resolving the $i$-th update at node $v$.
The algorithm is event-driven: if more than one events occur simultaneously, they are responded in the order listed in the algorithm.
The $i$-th proposed  update $\left(\ut{v}{i},\ppsl{v}{i}\right)$ at node $v$ is accepted once $\beta<\PAC$ and is rejected once $\beta\ge 1-\PRE$.
These two events are mutually exclusive because $\PAC + \PRE \leq 1$.
The two thresholds $\PAC,\PRE\in[0,1]$ are updated dynamically once $\mathbf{j}$ and $\mathbf{Y}$ updated correctly upon a message ``\MAC{}'' or ``\MRE{}''  received  from a neighbor $u\in N_v$.
Eventually at least one of the two events $\beta<\PAC$ and $\beta\ge 1-\PRE$ must occur when $|\CC{v}{i}|=1$, i.e.~when the correct configuration $\widehat{Y}_{\ut{v}{i}}(N_v)$ is fully known to $v$, because in this case $\PAC+\PRE=1$.
%

\begin{remark}[\textbf{proper coloring as special case}]
The subroutine \ResolveColoring{} (Algorithm~\ref{ResolveColoring}) is indeed a special case of the subroutine \Resolve{} (Algorithm~\ref{Resolve}) by setting  
\[
\forall \tau\in[q]^{N_v}:\quad f^v_{c,c'}(\tau) = \prod_{u \in N(v)}I[\tau_u \neq c'].
\] 
It is then easy to see that $\PAC\in\{0,1\}$ indicates the event $c' \notin \bigcup_{u \in N_v}\mathcal{S}_{\ut{v}{i}}(u)$ and $\PRE\in\{0,1\}$ indicates the event $\exists u \in N_v$ s.t. $\mathcal{S}_{\ut{v}{i}}(u) = \{c'\}$.
\end{remark}

\begin{remark}[\textbf{out-of-order resolution}]
The algorithm at each node $v$ resolves the updates in the order of the update times $0<\ut{v}{1}<\ut{v}{2}<\cdots<\ut{v}{m_v}<T$.
Alternatively, the algorithm can resolve any update once one of the conditions for that update in Line~\ref{alg:general-condition-1} and Line~\ref{alg:general-condition-2} in  Algorithm~\ref{Resolve} is triggered, disregarding the order of the update time.
To implement such \concept{out-of-order resolution} of updates, we need to slightly modify the construction of data structures $\mathbf{j}$ and $\mathbf{Y}$ at node $v$ to record respectively the indices and results for the resolved (out-of-order) updates of the neighbors $u\in N_v$.
The definition of the set $\mathcal{S}_t(u)$ of possible states for each neighbor $u\in N_v$ (Definition~\ref{definition-set-of-possible-values}) is changed accordingly to that supposed that $\ut{u}{j} \leq t <\ut{u}{j+1}$, $\mathcal{S}_{t}(u)\triangleq\{\,\Y{u}{j}\,\}$ if the $j$-th update of $u$ is resolved and its result is known to $v$; and if otherwise, $\mathcal{S}_{t}(u)\triangleq\{\,\Y{u}{j_0}\,\}\cup\left\{ \ppsl{u}{k} \mid j_0\le k\le j \right\}$, where $j_0$ is the index for the latest update of $u$ before the $j$-th update which is resolved and the result is known to $v$.
By coupling and a monotone argument, it is easy to verify that such implementation of the algorithm with out-of-order resolution of updates is at least as fast as the original algorithm with in-order resolution.

\end{remark}


%
%
\paragraph{Computation costs.}
The cost for local computation is dominated by the cost for computing the two thresholds $\PAC$ and $\PRE$, which are easy to compute for graphical models defined by edge factors, e.g.~Markov random fields, including all specific models mentioned in Section~\ref{sec:our-result}. 
For such graphical models, the Metropolis filter $f^v_{c,c'}(\tau)$ can be written as:
\begin{align*}
\forall \tau \in [q]^{N(v)}:\quad f^v_{c,c'}(\tau) = \min \Bigg\{1, \prod_{u \in N(v)}f^{v,u}_{c,c'}(\tau_u) \Bigg\}, \text{ where }f^{v,u}_{c,c'}:[q]\to\mathbb{R}_{\ge 0}.
\end{align*}
For this broad class of Metropolis filters, the thresholds $\PAC$ and $\PRE$ can be computed by the closed-forms:
\begin{align*}
\PAC &= \min \Bigg\{1, \prod_{u \in N_v} \left( \min_{b \in \mathcal{S}(u)} f^{v,u}_{c, c'}(b) \right) \Bigg\},\\
\PRE &=  1 - \min \Bigg\{1, \prod_{u \in N_v} \left( \max_{b \in \mathcal{S}(u)} f^{v,u}_{c, c'}(b) \right) \Bigg\},
\end{align*}
where $c=\widehat{Y}(v)$ and $c'=\ppsl{v}{i}$ as before and for each neighbor $u\in N_v$, the set $\mathcal{S}(u)=\mathcal{S}_{\ut{v}{i}}(u)$ is easy to compute locally based on the current $\mathbf{j}$ and $\mathbf{Y}$ as in~\eqref{eq-def-Stu}.

%% file: Analysis.tex
\section{Analysis of the Algorithm}
\label{sec:correctness}
We now analyze the \textbf{main algorithm} described in last section.

Recall the original continuous-time chain $(Y_t)_{t\in[0,T]}$ as defined in~\eqref{eq:eq-def-chain-Y}, and the continuous-time chain $(\widehat{Y}_t)_{t\in[0,T]}$ as defined in~\eqref{eq-def-chain-hat-Y}, generated by the algorithm.
%
%

\begin{theorem}[main theorem, revisited]\label{thm:main-theorem-Y}
Let $G=(V,E)$ be the network and $n=|V|$.
%
%
The followings hold for the main algorithm.
\begin{enumerate}
\item (correctness)
Upon termination, the algorithm outputs a random $\widehat{Y}_T\in[q]^V$, with each node $v\in V$ outputting $\widehat{Y}_T(v)$, such that $\widehat{Y}_T$ is identically distributed as $Y_T$.
\item (time complexity)
With high probability the algorithm terminates within $O(\Delta T+\log n)$ time units, where $\Delta$ denotes the maximum degree of $G$.

If Condition~\ref{condition-Lipschitz} holds, then
with high probability the algorithm terminates within $O(T+\log n)$ time units, where $O(\cdot)$ hides (linearly) the constant $C$ in Condition~\ref{condition-Lipschitz}.
\item (message complexity)
Each message contains $O(\log n+\log \lceil T\rceil+\log q)$ bits.
\end{enumerate}
\end{theorem}

The main theorem (Theorem~\ref{thm:main-theorem}) is consequence to above theorem and the following well-known concentration for Poisson distribution~\cite[Lemma 11]{hayes2013local}.

\begin{proposition}
\label{proposition-poisson}
Let $N \in \mathbb{Z}_{\geq 0}$ be a Poisson random variable with mean $\mu$, the following 	concentration inequalities hold for any $\epsilon < 1$:
\begin{align*}
\Pr[N \leq (1-\epsilon)\mu] &\leq \exp(-\epsilon^2\mu/2)	\\
\Pr[N \geq (1+\epsilon)\mu] &\leq \exp(-\epsilon^2\mu/3)	.
\end{align*}
Furthermore, if $t \geq 5\mu$, then
\begin{align*}
\Pr[N \geq t] \leq 2^{-t}.	
\end{align*}
\end{proposition}

%
Recall the single-site Metropolis chain $(X_t)_{t \ge 0}$ and the fact that for any $T>0$, $Y_T$ is identically distributed as $X_N$ for $N\sim\mathrm{Pois}(nT)$.
Let 
\begin{align*}
T' = 2T + 8\log n = O(T + \log n).	
\end{align*}
By Theorem~\ref{thm:main-theorem-Y}, assuming Condition~\ref{condition-Lipschitz}, there is a fully-asynchronous distributed algorithm that outputs $Y_{T'}$ within $O(T'+\log n)=O(T+\log n)$ time units with high probability.
 
Since $Y_{T'}$ is identically distributed as $X_{N}$ where $N\sim\mathrm{Pois}(nT')$, by the above concentration for Poisson distribution, we have
\begin{align*}
\Pr[N \geq Tn] \geq 1 - \frac{1}{n}.	
\end{align*}
Therefore, with high probability, the distributed algorithm outputs $X_N \in [q]^V$ for some $N \geq T n$ within $O(T + \log n)$ time units.
The main theorem (Theorem~\ref{thm:main-theorem}) is proved.
Theorem~\ref{thm:main-theorem-color} is implied by Theorem~\ref{thm:main-theorem} as a special case on proper graph coloring by setting every proposal distribution $\nu_v$ as the uniform distribution over $[q]$ and every Metropolis filter $f_{c,c'}^v:[q]^{N_v}\to[0,1]$ as $\forall \tau\in[q]^{N_v}: f^v_{c,c'}(\tau) = \prod_{u \in N(v)}I[\tau_u \neq c']$.

In the rest of the paper, we prove Theorem~\ref{thm:main-theorem-Y}. 

%% file: Appendix.tex
\subsection{Proof of correctness}
\label{section-proof-correct}


Recall that the algorithm defines a chain $(\widehat{Y}_t)_{t\in[0,T]}$ (formally defined in~\eqref{eq-def-chain-hat-Y}) and outputs $\widehat{Y}_T$ when terminates.
%
%
Now we show that this chain produced by the algorithm perfectly simulates the original continuous-time Metropolis chain $(Y_t)_{t\in[0,T]}$ (defined in~\eqref{eq:eq-def-chain-Y}). 

\begin{lemma}
\label{lemma-correctness}
With probability 1 the algorithm terminates and
$\widehat{Y}_T$ is identically distributed as $Y_T$.
\end{lemma}
Both processes $({Y}_t)_{t\in[0,T]}$ and $(\widehat{Y}_t)_{t\in[0,T]}$ are fully determined by the following randomnesses: 
\begin{align}
\begin{split}
m_v &\text{ for each }v\in V,\\ 
(\ut{v}{i},\ppsl{v}{i}, \bt{v}{i}) &\text{ for each }v\in V\text{ and every }1\le i\le m_v,
\end{split}\label{eq:correctness-randomness}
\end{align}
where $m_v$ denotes the number of times the Poisson clock at node $v$ rings, $\ut{v}{i}$ and $\ppsl{v}{i}$ denote respectively the time and proposal of the $i$-th update at node $v$, and $\bt{v}{i}\in[0,1]$ denotes the randomness used for the coin flipping in the $i$-th update of node $v$: in particular, in $({Y}_t)_{t\in[0,T]}$, the random value $\bt{v}{i}\in[0,1]$ is used in~\eqref{eq:eq-def-chain-Y} to determine the experiment of the Metropolis filter of the $i$-th update at node $v$, so that~\eqref{eq:eq-def-chain-Y} is implemented as
\begin{align*}
Y_{t}(v)=
\begin{cases}
\ppsl{v}{i} & \bt{v}{i}\le f_{c,c'}^{v}(Y_{t}(N_v))\\
Y_{t-\epsilon}(v) & \bt{v}{i}>f_{c,c'}^{v}(Y_{t}(N_v))
\end{cases},
\quad\text{where }c'=\ppsl{v}{i}\text{ and }c=Y_{t-\epsilon}(v);
\end{align*}
and in $(\widehat{Y}_t)_{t\in[0,T]}$, the random value $\bt{v}{i}$ is just the $\beta\in[0,1]$ in $\Resolve(i,\widehat{Y}(v), \mathbf{j}, \mathbf{Y})$ (Algorithm~\ref{Resolve}) for resolving the $i$-th update at node $v$.

\paragraph{The perfect coupling:} We couple the two processes $({Y}_t)_{t\in[0,T]}$ and $(\widehat{Y}_t)_{t\in[0,T]}$ by using the same randomnesses $(m_v)_{v\in V}$ and $(\ut{v}{i},\ppsl{v}{i}, \bt{v}{i})_{v\in V, 1\le i\le m_v}$.

We use the pair $(v,i)$, where $v\in V$ and $1\le i\le m_v$, to identify the $i$-th update at node $v$.
Fixed any collection of randomnesses as~\eqref{eq:correctness-randomness} such that all update times are distinct, a partial order $\prec$ on updates $(v,i)$ can be naturally defined as:
\begin{align}
\forall v\in V,  u\in N_v\cup\{v\},  1\le i\le m_u, 1\le j\le m_v: \quad (u,i)\prec(v,j) \iff \ut{u}{i}<\ut{v}{j}.\label{eq:correctness:poset}
\end{align}

Recall that $Y_{v}^{(i)}$ (defined in~\eqref{eq:def-step-Y}) denotes the state of $v$ in the chain $(Y_t)_{t\in[0,T]}$ right after the update $(v,i)$, and $\Y{v}{i}$ (constructed in Algorithm~\ref{ResolveAllUpdates}) denotes the state of $v$ in the algorithm after resolving the update $(v,i)$.

Assume that all $m_v$'s are finite and any pair of update times $\ut{u}{i}$ and $\ut{v}{j}$ for $(u,i)\neq(v,j)$ are distinct by a finite gap.
We then apply a structural induction to prove the following hypothesis:
\begin{claim}\label{claim:correctness-hypothesis}
For every $v\in V$ and every $1\le i\le m_v$, $\Y{v}{i}$ is well-defined and $\Y{v}{i}=Y_{v}^{(i)}$.
\end{claim}

Note that this proves Lemma~\ref{lemma-correctness}, because it implies that as long as all $m_v$'s are finite and all update times are distinct (which occurs with probability 1), 
$\widehat{Y}_T$ is well-defined, where $\widehat{Y}_T(v)=\Y{v}{m_v}$ for every $v\in V$ by~\eqref{eq-def-chain-hat-Y}, thus the algorithm terminates; 
and $\widehat{Y}_T$ is identically distributed as $Y_T$, where $Y_T$ is constructed as ${Y}_T(v)=Y_{v}^{(m_v)}$ for every $v\in V$ by~\eqref{eq:def-step-Y}.

It then remains to verify Claim~\ref{claim:correctness-hypothesis} by a structural induction on the partial order in~\eqref{eq:correctness:poset}.

In the \textbf{main algorithm}, each update time  is represented with bounded precision. 
After the truncation of precision, suppose each update time $t_u^i$ becomes $\widetilde{t}_u^i$.
The following property holds for truncated update times.
\begin{align}
\forall v\in V,  u\in N_v\cup\{v\},  1\le i\le m_u, 1\le j\le m_v: \widetilde{t}_u^i < \widetilde{t}_v^j \iff \ut{u}{i}<\ut{v}{j}.\label{eq:correctness:poset-truncate}
\end{align}
Hence, the partial order $\prec$ defined in~\eqref{eq:correctness:poset} is preserved after the truncation:
\begin{align}
\label{eq:order-truncate}
(u, i) \prec (v, j) \iff \widetilde{t}_u^i < \widetilde{t}_v^j.
\end{align}
\paragraph{Induction basis:}
We say the update $(v,j)$ is a minimal element with respect to partial order $\prec$ if and only if there is no update $(u, i)$ such that $(u, i) \prec (v, j)$. We prove that for each minimal element $(v, j)$, it holds that $\widehat{Y}_v^{(j)}$ is well-defined and $\widehat{Y}^{(j)}_v = Y^{(j)}_v$. Since $(v, j)$ is a minimal element, it must hold that $j = 1$, $\widetilde{t}^i_u > \widetilde{t}^j_v$ and $t^i_u > t^j_v$ for all $u \in N_v$ and $1\leq i \leq m_u$. 
In $\textbf{Phase II}$, node $v$ first tries to resolve the update $(v, j)$ and constructs the set $\mathcal{S}_{\widetilde{t}^j_v}(u) = \{\widehat{Y}_u^{(0)}\}$ for each $u \in N_v$ by~\eqref{eq-def-Stu}. Hence, once $v$ receives all $\widehat{Y}_u^{(0)}$ for $u \in N_v$, the update $(v, j)$ must be resolved. This implies $\widehat{Y}_v^{(j)}$ is well-defined. Note that $(v, j)$ is a minimal element with respect to partial order $\prec$. 
This implies $j = 1$, $Y_{t^j_v}(N_v) = Y_0(N_v)$ and $\widehat{Y}_{\widetilde{t}^j_v}(N_v) = \widehat{Y}_{0}(N_v)$.
In continuous-time Metropolis chain, it holds that $Y^{(j)}_v = c_v^j$ if $\beta_v^j < f^v_{c,c'}(Y_0(N_v))$; $Y^{(j)}_v = Y_v^{(0)}$ if $\beta_v^j \geq f^v_{c,c'}(Y_0(N_v))$, where $c = Y^{(0)}_v$ and $c' = c^j_v$. In \textbf{main algorithm}, the set $\CC{v}{j}$ defined in~\eqref{eq-def-CC} is $\{\widehat{Y}_0(N_v)\}$ and it holds that $\widehat{Y}^{(j)}_v = c_v^j$ if $\beta_v^j < f^v_{c,c'}(\widehat{Y}_0(N_v))$; $\widehat{Y}^{(j)}_v = \widehat{Y}_v^{(0)}$ if $\beta_v^j \geq f^v_{c,c'}(\widehat{Y}_0(N_v))$, where $c = \widehat{Y}^{(0)}_v$ and $c' = c^j_v$. Since $\widehat{Y}_u^{(0)} = Y_u^{(0)}$ for all $u \in V$, then $\widehat{Y}^{(j)}_v = Y^{(j)}_v$.

\paragraph{Induction step:}
Fix an update $(v, j)$. By induction hypothesis,  for all updates $(u, i)$ with $(u, i) \prec (v,j)$, it holds that $\widehat{Y}_u^{(i)}$ is well defined and $\widehat{Y}_u^{(i)} = Y_u^{(i)}$. For all $u \in N_v$ and $1\leq i \leq m_u$ with $(u, i) \prec (v, j)$, since $\widehat{Y}_u^{(i)}$ is well-defined, then in \textbf{main algorithm}, node $u$ must resolve the update $(u, i)$ and sends the message ``\MAC{}'' or ``\MRE''  to node $v$. Hence, node $v$ will eventually resolve the update $(v, j)$ no later than the moment at which all these messages are delivered, because in that moment, node $v$ knows the full information of $Y_{\widetilde{t}^j_v}(N_v)$ and $\PAC + \PRE = 1$.  This proves that $\widehat{Y}_v^{(j)}$ is well defined. 

Consider the continuous-time Metropolis chain. It holds that 
\begin{align}
\label{eq-con-chain-Yvj}
Y_v^{(j)} = \begin{cases}
c_v^j &\text{if }  \beta_v^j < f^v_{c,c'}(Y_{t^j_v}(N_v)) \text{ where } c = Y_v^{(j-1)}, c'=c_v^j\\
Y_v^{(j-1)} &\text{if }  \beta_v^j \geq f^v_{c,c'}(Y_{t^j_v}(N_v)) \text{ where } c = Y_v^{(j-1)}, c'=c_v^j.
 \end{cases}
\end{align}
Consider the moment at which $v$ resolves the update $(v, j)$. In Algorithm~\ref{Resolve}, node $v$ computes the set $\mathcal{S}_{\widetilde{t}^j_v}(u)$ by~\eqref{eq-def-Stu}, formally,
\begin{align}
\label{eq-Stu-proof}
\forall u \in N_v:\quad 
\mathcal{S}_{\widetilde{t}^j_v}(u) \triangleq \begin{cases}
\vspace{3pt}\left\{\,\Y{u}{k-1}\,\right\} 
{ \text{ where }\widetilde{t}_{u}^{k -1} \leq \widetilde{t}^j_v <\widetilde{t}_{u}^{k}}
&\text{if } 0 \leq \widetilde{t}^j_v < \widetilde{t}^{j_u}_u\\
\left\{\Y{u}{j_u-1}\right\} \cup \left\{ \ppsl{u}{k} \mid \widetilde{t}_{u}^{j_u} \leq \widetilde{t}_{u}^{k} \leq \widetilde{t}^j_v \right\}  &\text{if } \widetilde{t}_{u}^{j_u} \leq \widetilde{t}^j_v <T,
\end{cases}
\end{align}
where $j_u$  denotes that from $v$'s perspective, the neighbor $u$ is resolving its $j_u$-th update, and the outcomes of its first $(j_u-1)$ updates are known to $v$. In~\eqref{eq-Stu-proof}, if $0 \leq \widetilde{t}^j_v < \widetilde{t}^{j_u}_u$, let $\ell = k - 1$, where $\widetilde{t}_{u}^{k -1} \leq \widetilde{t}^j_v <\widetilde{t}_{u}^{k}$; if $\widetilde{t}_{u}^{j_u} \leq \widetilde{t}^j_v <T$, let $\ell = j_u - 1$. By~\eqref{eq:order-truncate}, it must hold that 
$(u,\ell) \prec (v, j)$.	
By induction hypothesis, we have $\widehat{Y}_u^{(\ell)} = Y_u^{(\ell)}$. By~\eqref{eq:correctness:poset-truncate}, the truncated update times preserve the order among all update times of node $u$ and node $v$. Combining with the update rule of the Metropolis chain, it holds that 
\begin{align*}
\forall u \in N_v:\quad
Y_{t^j_v}(u) \in \mathcal{S}_{\widetilde{t}^j_v}(u).
\end{align*}
Recall $\CC{v}{j} \triangleq \bigotimes_{u \in N_v}\mathcal{S}_{\widetilde{t}^j_v}(u)$. Then, we have
\begin{align*}
Y_{t^j_v}(N_v) \in \CC{v}{j}.	
\end{align*}
Recall that $\PAC \triangleq \min_{\tau \in \CC{v}{j}}f^v_{\widehat{c},c'}(\tau)$ and
$\PRE \triangleq \min_{\tau \in \CC{v}{j}}\left(1-f^v_{\widehat{c},c'}(\tau)\right)$ where $\widehat{c} = \widehat{Y}_v^{(j - 1)}$ and $c' = c_v^j$. Since $(v, j -1) \prec (v, j)$, then by induction hypothesis, we have $\widehat{Y}_v^{(j-1)}=Y_v^{(j-1)}$. Hence, the threshold $f^v_{c,c'}(Y_{t^j_v}(N_v))$ in~\eqref{eq-con-chain-Yvj} where $c = Y_v^{(j-1)} = \widehat{Y}_v^{(j-1)}$, satisfies 
\begin{align}
\label{eq-proof-relation}
\PAC \leq f^v_{c,c'}(Y_{t^j_v}(N_v))	 \quad\text{and}\quad \PRE \leq  1 - f^v_{c,c'}(Y_{t^j_v}(N_v)).
\end{align}
Finally, there are two cases when $v$ resolves the update $(v, j)$.
\begin{itemize}
\item Case $\beta_v^j < \PAC$: In \textbf{main algorithm}, it holds that $\widehat{Y}_v^{(j)} = c^j_v$.  	In continuous-time Metropolis chain, by~\eqref{eq-proof-relation}, it holds that $\beta_v^j < f^v_{c,c'}(Y_{t^j_v}(N_v))$, thus $Y_v^{(j)} = c^j_v$. This implies $\widehat{Y}_v^{(j)} = Y_v^{(j)}$.
\item Case $\beta_v^j \geq 1 - \PRE$: In \textbf{main algorithm}, it holds that $\widehat{Y}_v^{(j)} = \widehat{Y}_v^{(j-1)}$.  	In continuous-time Metropolis chain, by~\eqref{eq-proof-relation}, it holds that $\beta_v^j \geq  f^v_{c,c'}(Y_{t^j_v}(N_v))$, thus $Y_v^{(j)} = Y_v^{(j-1)}$. This implies $\widehat{Y}_v^{(j)} = Y_v^{(j)}$, because $\widehat{Y}_v^{(j-1)} = Y_v^{(j-1)}$ by induction hypothesis.
\end{itemize}
Combining two cases proves that $\widehat{Y}_v^{(j)} = Y_v^{(j)}$.

\subsection{Analysis of running time}
\label{sec:analysis}

We then upper bound the running time of the main algorithm.

\begin{lemma}
\label{lemma-convergence}
With high probability the main algorithm terminates within $O(\Delta T + \log n)$ time units. 
If Condition~\ref{condition-Lipschitz} holds, with high probability the algorithm terminates within $O(T + \log n)$ time units.
\end{lemma}

Note that each message of the main algorithm is of at most $O(\log n+\log\lceil T\rceil+\log q)$ bits (in fact, only 1 bit in \textbf{Phase~II}).
Theorem~\ref{thm:main-theorem-Y} is then implied by Lemma~\ref{lemma-correctness} and Lemma~\ref{lemma-convergence}.


We then proceed to prove Lemma~\ref{lemma-convergence}.
Given an execution of the algorithm, we define the \textbf{Phase I} of the execution as the time duration between the beginning of the algorithm and the moment at which the last node enters its \textbf{Phase II}, and define the  \textbf{Phase II} of the execution as the time duration between this moment and the termination of the algorithm.

Recall that for every node $v\in V$, the Poisson clock at $v$ rings $m_v$ times, where $m_v\sim\mathrm{Pois}(T)$.
By Proposition~\ref{proposition-poisson}, it holds that 
\[
\Pr[\,\forall v\in V, m_v\le 5T+3\log n\,]\ge 1- \frac{1}{n^2}.
\]
It also holds with high probability that no two update times are too close to each other (e.g.~within $<n^{-4}$ difference) since they are generated by rate-1 Poisson clocks, thus every update $(\ut{v}{i},\ppsl{v}{i})$ can be encoded in one message of $O(\log n+\log\lceil T\rceil+\log q)$ bits.
Therefore with high probability, the \textbf{Phase I} of the execution ends within $\max_{v\in V}m_v\le 5T+3\log n$ time units.

It then remains to upper bound the length of the \textbf{Phase II} of the execution.
We treat it as a standalone distributed algorithm, with different nodes starting asynchronously, while the complexity is measured starting from the moment the latest node starts.
\begin{definition}[\textbf{time complexity of a node in Phase II}]\label{def:complexity-phase-II}
%
For each node $v\in V$, we define the time for \concept{$v$ staying in \textbf{Phase II}} as the time duration between the earliest moment at which {all} nodes have entered their respective \textbf{Phase II} and the moment $v$ terminates.
\end{definition}

\begin{lemma}
\label{lemma-phase-II}
%
Fix any node $v\in V$.
With probability at least $1 - \frac{1}{n^2}$, node $v$ stays in \textbf{Phase II} for at most $O(\Delta T + \log n)$ time units.
And if Condition~\ref{condition-Lipschitz} holds, then with probability at least $1 - \frac{1}{n^2}$, node $v$ stays in \textbf{Phase II} for at most $O(T + \log n)$ time units.
\end{lemma}

Combined with the above analysis of \textbf{Phase I}, Lemma~\ref{lemma-phase-II} is sufficient to imply Lemma~\ref{lemma-convergence}.

The rest of the paper is dedicated to proving the upper bound for the \textbf{Phase II} of the algorithm in Lemma~\ref{lemma-phase-II}.
The proof is based on a notion of \concept{dependency chain}.

\section{The Dependency Chain}
\label{sec:running-time}

We introduce the notion of \concept{dependency chain}, for upper bounding the time complexity of  \textbf{Phase-II} algorithm (Algorithm~\ref{ResolveAllUpdates}). 

Fix all randomnesses of the main algorithm as defined in~\eqref{eq:correctness-randomness}: $m_v$ for each $v\in V$ and $(\ut{v}{i},\ppsl{v}{i},\bt{v}{i})$ for each $v\in V$ and every $1\le i\le m_v$, where all update times $\ut{v}{i}$ are distinct.


Recall that we use the pair $(v, i)$ to identify the $i$-th update  of node $v$. 
In Algorithm~\ref{Resolve}, the resolution of each update $(v, i)$ is triggered by one of the two \concept{resolution conditions} ($\beta<\PAC$ and $\beta\ge 1-\PRE$) respectively at Line~\ref{alg:general-condition-1}  and Line~\ref{alg:general-condition-2}.
This can be further divided into to cases:
\begin{itemize}
\item \textbf{self-triggered resolution:} a resolution condition is triggered by computing $\PAC$ and $\PRE$ for the first time in~Line~\ref{alg:general-pac-pre} of Algorithm~\ref{Resolve}; 
\item \textbf{resolution triggered by an adjacent update $(u, j)$:} a resolution condition is triggered by the recomputing of $\PAC$ and $\PRE$ upon processing the message ``\MAC'' or ``\MRE'' from a neighbor $u \in N_v$ that indicates the result for $u$ resolving its $j$-th update $(u,j)$.
\end{itemize}





Fixed all randomnesses, given an update $(v, i)$, we use $\DP{v}{i}$ to denote the {\em dependency chain that ends at $(v, i)$}. The dependency chain $\DP{v}{i}$ is a sequence of updates $(v_1,i_1),(v_2,i_2),\ldots,(v_{\ell},i_{\ell})$ where $(v_{\ell},i_{\ell})=(v,i)$, which is recursively constructed as:
\begin{itemize}
\item if the resolution of $(v, i)$ is a self-triggered resolution, then
\begin{align*}
\DP{v}{i} \triangleq \begin{cases}
(v, i) &\text{if } i = 1,\\
\DP{v}{i-1}, (v, i) &\text{if } i > 1;  	
 \end{cases}
\end{align*}
\item if the resolution of $(v, i)$ is triggered by the resolution of an adjacent update $(u, j)$, then 
\begin{align*}
\DP{v}{i}\triangleq \DP{u}{j},(v, i).
\end{align*}
\end{itemize}
Clearly such a dependency chain $\DP{v}{i}$ is uniquely constructed once all randomnesses are fixed.
%
Furthermore, any dependency chain satisfies the following monotonicity.

\begin{proposition}\label{prop:dependency-chain-monotone}
For any two consecutive updates $(u,j), (w,k)$ in a dependency chain $\DP{v}{i}$, it holds that $(u,j)\prec(w,k)$ where the partial order $\prec$ is as defined in~\eqref{eq:correctness:poset}.
\end{proposition}
\begin{proof}
First, if the resolution of $(w,k)$ is self-triggered, then $u=w$ and $j=k-1$, and hence $(u,j)\prec(w,k)$ clearly holds.
Otherwise, $w$ and $u$ are neighbors, and the resolution of $(w,k)$ is triggered by the resolution of $(u,j)$.
Then it must hold that $\ut{u}{j}<\ut{w}{k}$, since if otherwise,  by definition~\eqref{eq-def-Stu}, the construction of $\mathcal{S}_{\ut{w}{k}}(u)$ is unaffected by the resolution of the update $(u,j)$ at a later time $\ut{u}{j}>\ut{w}{k}$,
%
thus will not trigger any resolution condition of the update $(w,k)$ if the condition has not been satisfied already.
Therefore, it always holds that $(u,j)\prec(w,k)$.
%
%
%
%
\end{proof}


A direct consequence to the above proposition is that every dependency chain has finite length.



A key application of the dependency chain is that it can be used to upper bound the time complexity of Algorithm~\ref{ResolveAllUpdates}. 
\begin{proposition}
\label{proposition-dc-time}
%
Fix any node $v\in V$.
The number of time units for node $v$ staying in \textbf{Phase II}
is always upper bounded by the length of the dependency chain $\DP{v}{m_v}$, where $m_v$ denotes the total number of updates at node $v$.
\end{proposition}
\begin{proof}
Consider the main algorithm. Let $\mathcal{T}$ be the earliest moment at which all nodes have entered their respective \textbf{Phase II} and $\mathcal{T}'$ the moment at which $v$ terminates. 
The last update in dependency chain $\DP{v}{m_v}$ is  $(v, m_v)$, which must be resolved at the moment $\mathcal{T}'$. Let $(u, i)$ denote the first update in dependency chain $\DP{v}{m_v}$. Then $i = 1$. Suppose $(u, i)$ is resolved at the moment $\mathcal{T}''$. It must hold that $\mathcal{T}'' \leq \mathcal{T}$, because $(u, i)$ must be resolved once $u$ enters \textbf{Phase-II}. 
By definition~\ref{def:complexity-phase-II}, the time for $v$ staying in \textbf{Phase-II} is $\mathcal{T}'-\mathcal{T}$, which is upper bounded by $\mathcal{T}' - \mathcal{T}''$.
By the definition of the dependency chain $\DP{v}{m_v}$, it is easy to see that the number of time units in the time duration between $\mathcal{T}''$ and $\mathcal{T}'$ is at most the length of $\DP{v}{m_v}$.
\end{proof}

\subsection{Proof of Lemma~\ref{lemma-phase-II}}

Fix a node $v \in V$. 
Let $R_v$ denote the number of times units that node $v$ stays in \textbf{Phase II}. 
%

Fix an integer $\ell > 0$. 
We bound the probability that $R_v \geq \ell$.
By Proposition~\ref{proposition-dc-time} and Proposition~\ref{prop:dependency-chain-monotone}, if $R_v \geq \ell$, then there exist a sequence of nodes $v_1,v_1,\ldots,v_{\ell}\in V$  and a sequence of times $0<t_1<t_2<\cdots<t_\ell<T$ such that 
\begin{enumerate}[label=(\roman*)]
\item $v_{\ell}=v$ and $v_{j+1}\in N_{v_j}\cup\{v_j\}$ for every $1\le j\le \ell-1$;\label{item-1}
\item for every $1\le j\le \ell$, $t_j$ is an update time of $v_j$, and we denote by $k_j$ the order of this update time at $v_j$, i.e.~$t_j=\ut{v_j}{k_j}$;\label{item-2}
\item for every $2\le j\le \ell$, $(v_{j-1},k_{j-1})\prec(v_{j}, k_{j})$, and if $v_{j-1}\neq v_{j}$, the resolution of the update $(v_{j}, k_{j})$ is triggered by the resolution of the adjacent update $(v_{j-1},k_{j-1})$.\label{item-3}
\end{enumerate}
We call such a sequence of nodes $v_1,v_1,\ldots,v_{\ell}\in V$ a \concept{dependency path} if there exists a sequence of times $0<t_1<t_2<\cdots<t_\ell<T$  together with  $v_1,v_1,\ldots,v_{\ell}$ satisfying the above conditions.

We first prove the simple $O(\Delta T + \log n)$ upper bound. Fix a sequence of nodes $v_1,v_2,\ldots,v_{\ell}$ satisfying the condition in~\ref{item-1}. Consider the event that there exists a sequence of times $0<t_1<t_2<\cdots<t_\ell<T$ such that the Poisson clock at node $v_i$ rings at time $t_i$. By~\cite[Observation~3.2]{hayes2007general}, the probability of this event is at most $\left(\frac{\mathrm{e}T}{\ell} \right)^\ell$. A union bound over at most $(\Delta  + 1)^\ell$ possible paths satisfying the condition~\ref{item-1} implies
\begin{align*}
\Pr[\,R_v \geq \ell\,] \leq (\Delta + 1)^{\ell}\left(\frac{\mathrm{e}T}{\ell} \right)^\ell \leq \left(\frac{2\mathrm{e}T\Delta}{\ell} \right)^\ell.	
\end{align*}
Choosing $\ell =  4\mathrm{e}T\Delta + 2\log n $, we have
\begin{align*}
	\Pr[R_v \geq  4\mathrm{e}T\Delta + 2\log n] \leq \left(\frac{1}{2} \right)^{2\log n} = \frac{1}{n^2}.
\end{align*}
Hence, node $v$ stays in \textbf{Phase II} for at most $O(\Delta T + \log n)$ time units with probability at least $1- 1/n^2$. Remark that this proof does not need to use Condition~\ref{condition-Lipschitz}.


We now prove the $O(T + \log n)$ upper bound under Condition~\ref{condition-Lipschitz}.
Let $0\leq s < \ell$ be an integer.
We use  $\mathcal{P}(\ell, s)$ to denote the set of all sequences $ v_1,v_2,\ldots,v_\ell\in V$ satisfying that $v = v_\ell$, $v_{j+1} \in N(v_j) \cup \{v_j\}$ for every $1\leq j <\ell$, and $s = |\{1 \leq j \leq \ell-1 \mid v_j \neq v_{j+1} \}|$.



By the above argument, we have
\begin{align}
\label{eq-union-bound-path}
\Pr[\,R_v \geq \ell\,] \leq \sum_{s = 0}^{\ell - 1}	\sum_{P \in \mathcal{P}(\ell, s)}\Pr[\,P \text{ is a dependency path}\,].
\end{align}

We then bound the probability that a sequence $P$ is a dependency path.
Let random variable $\mathcal{N} \in \mathbb{Z}_{\geq 0}$ denote the total number of updates in the entire network before time $T$.
Apparently, $\mathcal{N}$ is given by the total number of times that $n$ rate-1 Poisson clocks ring up to time $T$ and follows the Poisson distribution with mean $nT$. 
We have
\begin{align}
\label{eq-partition-by-F}
\Pr[\,P \text{ is a dependency path}\,]	 &= \sum_{m \geq 0}\Pr[\,\mathcal{N} = m\,]\Pr[\,P \text{ is a dependency path} \mid \mathcal{N}=m\,]\notag\\
&=\mathrm{e}^{-nT} \sum_{m \geq 0}\frac{(nT)^m}{m!}\Pr[\,P \text{ is a dependency path} \mid \mathcal{N}=m\,].
\end{align}

Conditioning on $\mathcal{N}=m$ for a fixed integer $m\ge 0$, we define the following random variables:
\begin{align}
\label{eq-variables}
(U_1,T_1,C_1,\beta_1),(U_2,T_2,C_2,\beta_2),\ldots,(U_m,T_m,C_m,\beta_m),	
\end{align}
where each $U_i \in V$ is a random node, $0<T_1<T_2<\ldots<T_m<T$ are random times, each $C_i \in [q]$ is a random proposal distributed as $\nu_{U_i}$, and each $\beta_i \in [0,1)$ is uniformly distributed over $[0,1)$, such that:
\begin{itemize}
\item the Poisson clock at node $U_i$ rings at time $T_i$;
\item node $U_i$ proposes $C_i$ for its update at time $T_i$;
\item node $U_i$ samples the random real number $\beta_i \in [0,1)$ at Line~\ref{alg:general-sample} of Algorithm~\ref{Resolve}  to resolve its update at time $T_i$.
\end{itemize}

We now use $\UD_k$ for $1\leq k \leq m$ to identify the update of node $U_k$ at time $T_k$ with proposal $C_k$.


Conditioning on $\mathcal{N} = m$, if $P=v_1,v_2,\ldots,v_{\ell}$ is a dependency path, then from the above discussion we know that there exist $\ell$ indices $1\leq p(1)<p(2)<,...,<p(\ell) \leq m$ such that the following events occur simultaneously:
\begin{itemize}
\item \textbf{event} $\mathcal{A}_1$: for all $1\leq j \leq \ell$, $U_{p(j)} = v_j$;
\item \textbf{event} $\mathcal{A}_2^{(j)}$, where $2\leq j \leq \ell$: either $U_{p(j-1)} = U_{p(j)}$ or the resolution of $\UD_{p(j)}$ is triggered by the resolution of the adjacent update $\UD_{p(j-1)}$.


\end{itemize}

Fix any $\ell$ indices $1\leq p(1)<p(2)<,...,<p(\ell) \leq m$. 
We bound the following probability
\begin{align}
\label{eq-prob-A1-A2}
&\Pr\left[\,\mathcal{A}_1 \wedge \left(\bigwedge_{j=2}^{\ell} \mathcal{A}_2^{(j)}\right) \mid \mathcal{N}=m \,\right]\notag\\	
=&\Pr[\,\mathcal{A}_1\mid \mathcal{N}=m\,]\prod_{j=2}^{\ell}\Pr\left[\,\mathcal{A}^{(j)}_2 \mid \mathcal{N}=m \wedge \mathcal{A}_1 \wedge \left(\bigwedge_{k=2}^{j-1} \mathcal{A}_2^{(k)}\right)\,\right].
\end{align}
%
Note that conditioning on $\mathcal{N} = m$, each $U_i$ is uniformly and  independently distributed in  $V$.
This can be proved by an alternative equivalent process: there is a single rate-$n$ Poisson clock, and once the clock rings, a node is picked uniformly at random. 
Then conditioning on the rate-$n$ Poisson clock rings for $m$ times, each $U_i$ is uniformly and independently distributed in $V$.
Therefore,
\begin{align}
\label{eq-prob-A1}
\Pr[\,\mathcal{A}_1\mid \mathcal{N}=m\,] = \left(\frac{1}{n}\right)^\ell.
\end{align}
We further make the following claim on events $\mathcal{A}_2^{(j)}$.
\begin{claim}
\label{claim-conditional-prob}
Assume that Condition~\ref{condition-Lipschitz} holds.
For any $2\leq j \le\ell$ where $v_j \neq v_{j-1}$, it holds that
\begin{align*}
&\Pr\left[\,\mathcal{A}^{(j)}_2 \mid \mathcal{N}=m \wedge \mathcal{A}_1 \wedge \left(\bigwedge_{k=2}^{j-1} \mathcal{A}_2^{(k)}\right)\,\right] \leq \frac{2C}{\Delta},
\end{align*}	
where $C$ is the constant in Condition~\ref{condition-Lipschitz}.
\end{claim}

Since $P \in \mathcal{P}(\ell, s)$, there are exactly $s$ indices $j$ such that   $v_j \neq v_{j+1}$. Combining\eqref{eq-prob-A1-A2},~\eqref{eq-prob-A1} and Claim~\ref{claim-conditional-prob} yields
\begin{align*}
\Pr\left[\,\mathcal{A}_1 \wedge \left(\bigwedge_{j=2}^{\ell} \mathcal{A}_2^{(j)}\right) \mid \mathcal{N}=m \,\right] \leq \left(\frac{1}{n}\right)^\ell\left(\frac{2C}{\Delta}\right)^s.
\end{align*}
Taking a union bound over $\binom{m}{\ell}$ possible indices $1\leq p(1)<p(2)<,...,<p(\ell) \leq m$ yields
\begin{align*}
\Pr[\,P \text{ is a dependency path} \mid \mathcal{N}=m\,] \leq \binom{m}{\ell}\left(\frac{1}{n}\right)^\ell\left(\frac{2C}{\Delta}\right)^s.
\end{align*}
Finally, note that $|\mathcal{P}(\ell, s)| \leq \binom{\ell-1}{s}\Delta^s$. Together with~\eqref{eq-union-bound-path} and~\eqref{eq-partition-by-F}  it gives that
\begin{align*}
\Pr[R_v \geq \ell] &\leq \sum_{s = 0}^{\ell - 1}	\sum_{P \in \mathcal{P}(\ell, s)}\mathrm{e}^{-nT}\sum_{m = 0}^{\infty}\frac{(nT)^m}{m!}\binom{m}{\ell}\left(\frac{1}{n}\right)^\ell\left(\frac{2C}{\Delta}\right)^s\\
&\leq  \sum_{s = 0}^{\ell - 1}\binom{\ell - 1}{s}\Delta^{s}\mathrm{e}^{-nT}\sum_{m = \ell}^{\infty}\frac{(nT)^m}{m!}\binom{m}{\ell}\left(\frac{1}{n}\right)^\ell\left(\frac{2C}{\Delta}\right)^s\\
&\leq \frac{T^\ell}{\ell!}\left(1 + 2C\right)^\ell\\
&\leq \left( \frac{T\mathrm{e}(1+2C)}{\ell} \right)^\ell.
\end{align*}
Choosing $\ell =  2\mathrm{e}\left(1 + 2C\right)T + 2\log n $, we have
\begin{align*}
	\Pr[R_v \geq  2\mathrm{e}\left(1 + 2C\right)T + 2\log n] \leq \left(\frac{1}{2} \right)^{2\log n} = \frac{1}{n^2}.
\end{align*}
Hence, node $v$ stays in \textbf{Phase II} for at most $O(T + \log n)$ time units with probability at least $1- 1/n^2$.
This proves Lemma~\ref{lemma-phase-II}.

\begin{proof}[Proof of Claim~\ref{claim-conditional-prob}]
By the definition of fully-asynchronous message-passing model in Section~\ref{section-model}, all the message delays are determined by an {adversarial scheduler} who is adaptive to the entire input. Given the input of each node, we can fix all the message delays in \textbf{main algorithm}. Denote the delays of all messages as $\mathcal{F}_D$. Note that $\mathcal{F}_D$ satisfies the constraint that each unidirectional channel is a reliable FIFO channel.

Recall that $\mathcal{N}$ is the total number of updates in the entire network before time $T$. Consider the random variables defined in~\eqref{eq-variables}:
\begin{align}
\label{eq-update-seq}
(U_1,T_1,C_1,\beta_1),(U_2,T_2,C_2,\beta_2),\ldots,(U_\mathcal{N},T_\mathcal{N},C_\mathcal{N},\beta_\mathcal{N}),
\end{align}
where $T_1<T_2<\ldots<T_\mathcal{N}$. We first fix the randomness of all Poisson clocks:
\begin{itemize}
\item $\mathcal{F}_1$: fix $\mathcal{N} = m$ and fix the values of all $U_1,U_2,\dots,U_m$ and $T_1, T_2,\ldots,T_m$.	
\end{itemize}
Recall $P=v_1,v_2,\ldots,v_\ell$ is the fixed path and $1\leq p(1)<p(2)<\ldots<p(\ell) \leq m$ is the $\ell$ fixed indices.
Fix an integer $2\leq j \leq\ell$ such that $v_{j-1} \neq v_{j}$.
We then fix the randomness of the first $p(j) - 1$ updates in sequence~\eqref{eq-update-seq}:
\begin{itemize}
\item  $\mathcal{F}_2$: fix the values of all $C_k, \beta_k$ for $1\leq k \leq p(j) - 1$.
\end{itemize}

We claim that the following two results hold.
\begin{enumerate}[label=(R\arabic*)]
\item Given any $\mathcal{F}_D$, $\mathcal{F}_1$ and $\mathcal{F}_2$, it holds that $\mathcal{N} = m$ and the occurrences of events $\mathcal{A}_1$ and $\mathcal{A}_2^{(k)}$ for all $2\leq k \leq j-1$ are fully determined.	\label{result-A}
\item For any $\mathcal{F}_D$, $\mathcal{F}_1$ and $\mathcal{F}_2$ under which $\mathcal{A}_1$ and all $\mathcal{A}_2^{(k)}$ for $2\leq k \leq j-1$ occur,
if Condition~\ref{condition-Lipschitz} is satisfied, then: 
\begin{align*}
\Pr[\,\mathcal{A}_2^{(j)} \mid \mathcal{F}_D \land\mathcal{F}_1\land\mathcal{F}_2\,] \leq \frac{2C}{\Delta},	
\end{align*} 
where the probability takes over the randomness of unfixed variables $C_{p(j)}$ and $\beta_{p(j)}$.\label{result-B}
\end{enumerate}
The Claim~\ref{claim-conditional-prob} is proved by combining above two results.

We first prove result~\ref{result-A}.
It is easy to see that $\mathcal{F}_1$ determines whether the event $\mathcal{A}_1$ occurs.
We prove that the occurrences of events $\mathcal{A}_2^{(k)}$ for all $2\leq k \leq j-1$ are determined. Given $\mathcal{F}_1$ and $\mathcal{F}_2$, the evolution of the continuous-time Metropolis chain $Y_t$ from $t = 0$ to $t = T_{p(j)}-\epsilon$ is fully determined. Given $\mathcal{F}_D,\mathcal{F}_1$ and $\mathcal{F}_2$, define algorithm $\mathcal{A}$ as a modified version of the \textbf{main algorithm} such that for each node $v \in V$:
\begin{itemize}
\item in \textbf{Phase I}, node $v$ generates all the update times and random proposals conditioning on $\mathcal{F}_1$ and $\mathcal{F}_2$, then exchanges this information with neighbors as the \textbf{main algorithm}.
\item in \textbf{Phase II}, node $v$ resolves each update as the \textbf{main algorithm} and $v$ samples $\beta$ in Line~\ref{alg:general-sample} of Algorithm~\ref{Resolve} conditioning on $\mathcal{F}_2$;  once node $v$ needs to use any variable whose value is not fixed by $\mathcal{F}_1$ and $\mathcal{F}_2$, then the algorithm $\mathcal{A}$ at node $v$ terminates immediately.	
\end{itemize}

By a similar induction argument in Section~\ref{section-proof-correct}, it can be verified that algorithm $\mathcal{A}$ simulates the  continuous-time Metropolis chain up to time $T_{p(j)}-\epsilon$ and generates $(\widehat{Y}_t)_{t \in [0, T_{p(j)})}$.
The algorithm $\mathcal{A}$ resolves all the updates $(U_k, T_k, C_k, \beta_k)$ for $1 \leq k \leq p(j) - 1$ in~\eqref{eq-update-seq}.
In algorithm~$\mathcal{A}$, for each node $v \in V$, the \textbf{Phase II} of node $v$ is fully determined by $\mathcal{F}_D, \mathcal{F}_1$ and $\mathcal{F}_2$.
For each $2\leq k \leq j-1$, by the definition of event $\mathcal{A}_2^{(k)}$, its occurrence is fully determined given $\mathcal{F}_D,\mathcal{F}_1$ and $\mathcal{F}_2$.

We then prove result~\ref{result-B}. Now, we only consider $\mathcal{F}_D$, $\mathcal{F}_1$ and $\mathcal{F}_2$ under which $\mathcal{A}_1$ and all $\mathcal{A}_2^{(k)}$ for $2\leq k \leq j-1$ occur.
Hence, we assume $U_{p(j-1)}=v_{j-1}$ and $U_{p(j)} = v_j$. Recall that $v_j \neq v_{j-1}$ and 
 $\mathsf{UD}_k$ denotes the $k$-th update $(U_k, T_k, C_k, \beta_k)$ in~\eqref{eq-update-seq}.
The event $\mathcal{A}_2^{(j)}$ occurs if and only if the update $\UD_{p(j)}$ is triggered by the resolution of the adjacent update $\UD_{p(j-1)}$.
In Algorithm~\ref{Resolve}, to resolve the update $\UD_{p(j)}$, node $v_j$ needs to compute the set $\mathcal{S}_{T_{p(j)}}(w)$ defined in~\eqref{eq-def-Stu} for all $w \in N(v_j)$.
For any $1\leq k\leq m$ and $w \in N(U_k)$, let $\Msg{k}{w}$ denote the message sent from $U_k$ to $w$ that indicates whether the update $\UD_k$ is accepted.
%
When computing the set $\mathcal{S}_{T_{p(j)}}(w)$ for $w \in N(v_j)$, node $v_j$ only uses the following information:
\begin{enumerate}[label=(I\arabic*)]	
\item $\{T_k \mid 1\leq k \leq m \land U_k = w\} \cup \{T_{p(j)}\}$\label{infa-t};
\item $\{C_k \mid k \leq p(j) - 1 \land U_k = w\} \cup \{\widehat{Y}_0(w)\}$;\label{info-a}
\item all messages $\Msg{k}{v_j}$ received by $v_j$ satisfying $k \leq p(j)-1$ and $U_k = w$.\label{info-b}
\end{enumerate}
By the definition in~\eqref{eq-def-Stu}, node $v_j$ computes $\mathcal{S}_{T_{p(j)}}(w)$ based on the current $j_w$ and $(\widehat{Y}_w^{(j)})_{0 \leq j\leq j_w}$, where $j_w$ and $(\widehat{Y}_w^{(j)})_{0 \leq j\leq j_w}$ are computed according to the messages received from $w$. Note that computing $\mathcal{S}_{T_{p(j)}}(w)$  only needs to use the information about the updates at node $w$ whose update times are before $T_{p(j)}$, because  $\mathcal{S}_{T_{p(j)}}(w)$ defined in~\eqref{eq-def-Stu} is the set of possible states for $w$ at time $T_{p(j)}$. Also note that in \textbf{main algorithm}, all update times are represented with bounded precision satisfying~\eqref{eq:correctness:poset-truncate}, the partial order $\prec$ in~\eqref{eq:correctness:poset} is preserved. Thus, node $v_j$ only uses information in~\ref{infa-t},~\ref{info-a}, and~\ref{info-b} to compute $\mathcal{S}_{T_{p(j)}}(w)$.

Given $\mathcal{F}_D$, $\mathcal{F}_1$ and $\mathcal{F}_2$, the sets in~\ref{infa-t} and~\ref{info-a} are fixed and node $v_j$ knows these sets once $v_j$ enters the \textbf{Phase-II}. Define the set of messages
\begin{align*}
\mathcal{M}= \left\{ \Msg{k}{v_j} \mid k \leq p(j)-1 \land U_k \in N(v_j) \right\}	.
\end{align*}
According to the proof of result~\ref{result-A}, we know that for each message $M \in \mathcal{M}$, the content of $M$ and the moment at which node $v_j$ processes $M$ is fully determined given $\mathcal{F}_D$, $\mathcal{F}_1$ and $\mathcal{F}_2$.
%
Consider the moment $\mathcal{T}$ at which node $v_j$ processes the message $\Msg{p(j-1)}{v_j}\in \mathcal{M}$ to update the current $(\boldsymbol{Y}, \mathbf{j})$. Suppose after the update, the pair $(\boldsymbol{Y}, \mathbf{j})$ becomes  $(\boldsymbol{Y}', \mathbf{j}')$.
For any $w \in N(v_j)$, let $\mathcal{S}(w)$ be the set $\mathcal{S}_{T_{p(j)}}(w)$ computed by $v_j$ according to~\eqref{eq-def-Stu} based on $(\boldsymbol{Y}, \mathbf{j})$; and let 
$\mathcal{S}'(w)$ be the set $\mathcal{S}_{T_{p(j)}}(w)$ computed by $v_j$ according to~\eqref{eq-def-Stu} based on $(\boldsymbol{Y}', \mathbf{j}')$. Given $\mathcal{F}_D$, $\mathcal{F}_1$ and $\mathcal{F}_2$, it holds that 

\begin{itemize}
\item for all $w \in N(v_j)$, both $\mathcal{S}(w)$ and $\mathcal{S}'(w)$ are fixed,
\end{itemize}
furthermore, it holds that $U_{p(j-1)} = v_{j - 1}$ and
\begin{enumerate}[label=(P\arabic*)]
\item for all $w \in N(v_j) \setminus\{v_{j-1}\}$, $\mathcal{S}(w) = \mathcal{S}'(w)$;\label{P1}
\item $\mathcal{S}'(v_{j-1}) \subseteq \mathcal{S}(v_{j-1})$ and $|\mathcal{S}(v_{j-1})| - |\mathcal{S}'(v_{j-1})|\leq 1$;\label{P2}
\item for all $w \in N(v_j)$, $|\mathcal{S}'(w)| \geq 1$.\label{P3}
\end{enumerate}
Properties~\ref{P1} and~\ref{P2} hold because $v_j$ only processes one the message $\Msg{p(j-1)}{v_j}$ sent from $v_{j-1}$ at moment $\mathcal{T}$. Property~\ref{P3} holds because the set $\mathcal{S}_{T_{p(j)}}(w)$ contains at least one element due to the definition in~\eqref{eq-def-Stu}.

Let $c$ denote $\widehat{Y}_t(v_j)$ where $t = T_{p(j)}-\epsilon$. Note that $c$ is fixed given $\mathcal{F}_D$, $\mathcal{F}_1$ and $\mathcal{F}_2$.
Let $c'$ denote the proposal $C_{p(j)}$.
Let $\PAC, \PRE$ be the acceptance  and rejection thresholds for update $\UD_{p(j)}$ obtained from $(\mathcal{S}(w))_{w \in N(v_j)}$ and  $\PAC', \PRE'$ the thresholds obtained from $(\mathcal{S}'(w))_{w \in N(v_j)}$, formally
\begin{align*}
\PAC &= \min_{\tau \in \mathcal{C}}f^{v_j}_{c, c'}(\tau) & \PRE &= 1 - \max_{\tau \in \mathcal{C}}f^{v_{j}}_{c, c'}(\tau)\\
\PAC' &= \min_{\tau \in \mathcal{C}'}f^{v_j}_{c, c'}(\tau) & \PRE' &= 1 - \max_{\tau \in \mathcal{C}'}f^{v_j}_{c, c'}(\tau),
\end{align*}
where $\mathcal{C} = \bigotimes_{w \in N(v_j)}\mathcal{S}(w)$ and  $\mathcal{C}' = \bigotimes_{w \in N(v_j)}\mathcal{S}'(w)$.

If the event $\mathcal{A}_2^{(j)}$ occurs, then the following event $\mathcal{B}_j$ must occur.
\begin{itemize}
\item \textbf{event} $\mathcal{B}_j$:  $ (\PAC \leq \beta_{p(j)} < 1 -  \PRE) \land (\beta_{p(j)}< \PAC' \lor \beta_{p(j)} \geq 1 - \PRE')$.
\end{itemize}
Suppose the event $\mathcal{A}_2^{(j)}$ occurs. Node $v_j$ resolves the update $\UD_{p(j)}$ right after  moment $\mathcal{T}$, which implies  $\beta_{p(j)}< \PAC' \lor \beta_{p(j)} \geq 1 - \PRE'$.
And node $v_j$ cannot resolve the update $\UD_{p(j)}$ right before moment $\mathcal{T}$, which implies $ \PAC\leq \beta_{p(j)} < 1 -  \PRE$.  

Given $\mathcal{F}_D$, $\mathcal{F}_1$ and $\mathcal{F}_2$,  note that $c' = C_{p(j)}$ is an independent random proposal from distribution $\nu_{v_j}$ and $\beta_{p(j)}$ is an independent random real number uniformly distributed over $[0,1)$. Then 
\begin{align}
\label{eq-bound-a2j}
\Pr[\,\mathcal{A}_2^{(j)} \mid \mathcal{F}_D \land\mathcal{F}_1\land\mathcal{F}_2\,]	&\leq \Pr[\,\mathcal{B}_j\mid \mathcal{F}_D \land\mathcal{F}_1\land\mathcal{F}_2\,]\notag\\
&\leq \EE{\,c' \sim \nu_{v_j} }{(\PAC' - \PAC) + (\PRE' - \PRE) \mid\mathcal{F}_D \land\mathcal{F}_1\land\mathcal{F}_2\,},
\end{align}
where the last inequality holds because $\beta_{p(j)}$ is uniformly distributed over $[0,1)$,  $\PAC' \geq \PAC$ and $\PRE' \geq \PRE$ (because $\mathcal{S}'(w) \subseteq \mathcal{S}(w)$ for all $w \in N(v_j)$). 

Finally, we bound the expectation in~\eqref{eq-bound-a2j}. Given $\mathcal{F}_D$, $\mathcal{F}_1$ and $\mathcal{F}_2$, the value $c = \widehat{Y}_{T_{p(j)} - \epsilon}(v_j)$ and all the sets $\mathcal{S}(w)$ and $\mathcal{S}'(w)$ are determined, and the Properties~\ref{P1},~\ref{P2} and~\ref{P3} hold. Recall
\begin{align*}
(\PAC' - \PAC) + (\PRE' - \PRE) = \min_{\tau \in \mathcal{C}'}f^{v_j}_{c, c'}(\tau) -  \min_{\tau \in \mathcal{C}}f^{v_j}_{c, c'}(\tau) + 	\max_{\tau \in \mathcal{C}}f^{v_{j}}_{c, c'}(\tau) - \max_{\tau \in \mathcal{C}'}f^{v_j}_{c, c'}(\tau).
\end{align*}
We introduce the following optimization problem to find the maximum value of the expectation in~\eqref{eq-bound-a2j} under the worst case of $\mathcal{F}_D$, $\mathcal{F}_1$ and $\mathcal{F}_2$.
Fix two nodes $\{v, u\} \in E$, define the optimization problem $\mathfrak{P}(v,u)$  as follows. 
\vspace{7pt}
\begin{center}
\fbox{
\parbox[c][][c]{38em}{
\vspace{-7pt}
\begin{align}
\label{eq-def-opt}
	\text{variables}\quad  &S_1(w) \subseteq [q], S_2(w) \subseteq [q] &\forall w \in N_v\\
						   &c \in [q]&\notag\\
	\text{maximize}\quad& \sum_{c' \in [q]}\nu_v(c')\left(\min_{\tau \in \mathcal{C}_2}\F(\tau) - \min_{\tau \in \mathcal{C}_1}\F(\tau) + \max_{\tau \in \mathcal{C}_1}\F(\tau) - \max_{\tau \in \mathcal{C}_2}\F(\tau)\right) \span \span \notag\\
	\text{subject to}\quad& \mathcal{C}_1 = \bigotimes_{w \in N_v}S_1(w),\quad \mathcal{C}_2 = \bigotimes_{w \in N_v}S_2(w)&\notag\\
					 &S_2(u) \subset S_1(u) \notag\\
					 &|S_1(u)|-|S_2(u)| =1 \notag\\
					 &|S_2(w)| \geq 1  &\forall w \in N_v \notag\\
					 &S_2(w) = S_1(w) &\forall w \in N_{v} \setminus \{u\}\notag
\end{align}
\vspace{-15pt}
}
}
\end{center}
\vspace{7pt}

Remark that we use constraint $|S_1(u)| -|S_2(u)| =1$ rather than  $|S_1(u)| -|S_2(u)| \leq 1$ as Property~\ref{P2} because the value of the objective function is 0 if $|S_1(u)| =|S_2(u)| $.

We claim the  optimal value of the objective function in problem $\mathfrak{P}(v, u)$  is at most 
$$2\max_{a, b, c\in [q]} \EE{c' \sim \nu_v}{\delta_{u, a, b}\F}.$$ 
This result is proved in Section~\ref{section-prove-opt-lemma}.
The expectation in~\eqref{eq-bound-a2j} is upper bounded by the optimal value of the objective function in problem $\mathfrak{P}(v_{j}, v_{j- 1})$. 
By Condition~\ref{condition-Lipschitz}, we have
\begin{align*}
\Pr[\,\mathcal{A}_2^{(j)} \mid \mathcal{F}_D \land \mathcal{F}_1\land \mathcal{F}_2\,] \leq  2\max_{a, b, c\in [q]} \EE{\,c' \sim \nu_{v_j}}{\delta_{v_{j-1}, a, b}f^{v_j}_{c,c'}\,} \leq \frac{2C}{\Delta}.	
\end{align*} 
\end{proof}


\subsection{Analysis of the optimization problem}
\label{section-prove-opt-lemma}
\begin{lemma}
\label{lemma-opt-problem}
Fix an edge $\{v,u\}\in E$.
Let $\OPT$ denote the objective function value of the optimal solution to problem $\mathfrak{P}(v,u)$ defined in~\eqref{eq-def-opt}. It holds that
\begin{align*}
\OPT \leq 	2\max_{a, b, c\in [q]} \EE{c' \sim \nu_v}{\delta_{u, a, b}\F}.
\end{align*}
\end{lemma}
\begin{proof}
Suppose we replace the objective function in problem $\mathfrak{P}(v,u)$ with 
\begin{align}
\label{eq-opt-1}
\text{maximize}\quad \sum_{c' \in [q]}\nu_v(c')\left(\min_{\tau \in \mathcal{C}_2}\F(\tau) - \min_{\tau \in \mathcal{C}_1}\F(\tau)\right)
\end{align}
and keep all variables and constraints unchanged. We obtain a new optimization problem.
Let $\OPT_1$ denote the objective function value of the optimal solution to this problem. 

Similarly, suppose we replace the objective function in problem $\mathfrak{P}(v,u)$ with
\begin{align*}
\text{maximize}\quad \sum_{c' \in [q]}\nu_v(c')\left(\max_{\tau \in \mathcal{C}_1}\F(\tau) - \max_{\tau \in \mathcal{C}_2}\F(\tau)\right)
\end{align*}
and keep all variables and constraints unchanged. We obtain another new optimization problem.
Let $\OPT_2$ denote the objective function value of the optimal solution to this problem. 

It is easy to verify 
\begin{align*}
\OPT \leq \OPT_1 + \OPT_2.	
\end{align*}
We show that
\begin{align}
\OPT_1	&\leq \max_{a, b, c\in [q]} \EE{c' \sim \nu_v}{\delta_{u, a, b}\F} \label{eq-opt-bound-1}\\
\OPT_2 &\leq \max_{a, b, c\in [q]} \EE{c' \sim \nu_v}{\delta_{u, a, b}\F}.\label{eq-opt-bound-2}
\end{align}
This proves the lemma.

We prove inequality~\eqref{eq-opt-bound-1}. Inequality~\eqref{eq-opt-bound-2} can be proved by going through a similar proof.

Consider the new optimization problem with objective function~\eqref{eq-opt-1}. We claim the following result for this problem.
\begin{claim}
\label{claim-optsol}
There exists an optimal solution $\SOL^\star = (\boldsymbol{S}_1^\star, \boldsymbol{S}_2^\star , c^\star)$ such that $|S^\star_2(u)| = 1$, where  $\boldsymbol{S}_1^\star = (S_1^\star(w))_{w \in N(v)}$ and $\boldsymbol{S}_2^\star = (S_2^\star(w))_{w \in N(v)}$.
\end{claim}
Thus, we have $|S^\star_1(u)| = 2$ due to the constraint of the problem. Suppose $S^\star_1(u)=\{a, b\}$ and $S^\star_2(u) = \{b\}$. Fix a value $c' \in [q]$, define
\begin{align*}
\gamma_1(c') &= \min_{\tau \in \mathcal{C}^\star_1}\FC(\tau) = \FC(\tau')\\
\gamma_2(c') &= \min_{\tau \in \mathcal{C}^\star_2}\FC(\tau) = \FC(\tau''),
\end{align*}
where
\begin{align*}
\mathcal{C}_1^\star = \bigotimes_{w \in N_v}S^\star_1(w),\quad \mathcal{C}_2^\star = \bigotimes_{w \in N_v}S^\star_2(w),	
\end{align*}
and $\tau' = \arg\min_{\tau \in \mathcal{C}^\star_1}\FC(\tau), \tau'' = \arg \min_{\tau \in \mathcal{C}^\star_2}\FC(\tau)$.
It must hold that $\tau''_u = b$ because $S^\star_2(u) = \{b\}$. There are two cases for $\tau'_u$: $\tau'_u = a$ or $\tau'_u = b$, because $S^\star_1(u) = \{a, b\}$. 

Suppose $\tau'_u = b$. Since  $S_1^\star(w) = S_2^\star(w)$ for all $w \in N_v \setminus \{u\}$, then we must have 
\begin{align*}
\gamma_2(c') - \gamma_1(c')	= 0 \leq \delta_{u,a,b}\FC.
\end{align*}

Suppose $\tau'_u = a$. We define $\tau''' \in [q]^{N_v}$ as
\begin{align*}
\tau'''_w = \begin{cases}
b &\text{if } w = u\\
\tau'_w &\text{if } w \neq u. 	
 \end{cases}
\end{align*}
Note that $b \in S_2^\star(u)$ and $\tau'_w  \in S_2^\star(w)$ for all $w \in N_v \setminus \{u\}$ (because $S_2^\star(w) = S_1^\star(w)$). We have $\tau''' \in \mathcal{C}_2^\star$, which implies $\FC(\tau''') \geq \FC(\tau'')$. Hence 
\begin{align*}
\gamma_2(c') - \gamma_1(c')	= \FC(\tau'') - \FC(\tau') \leq \FC(\tau''') - \FC(\tau') \leq  \delta_{u,a,b}\FC.
\end{align*}
The last inequality is because $\tau'$ and $\tau'''$ agree on all nodes except $u$ and $\tau'_u = a, \tau'''_u =b$.

Combining above two cases together, we have 
\begin{align*}
\OPT_1&=\sum_{c' \in [q]}\nu_v(c')\left(\min_{\tau \in \mathcal{C}_2^\star}\FC(\tau) - \min_{\tau \in \mathcal{C}_1^\star}\FC(\tau)\right)\\
& =\sum_{c' \in [q]}\nu_v(c') \left( \gamma_2(c') - \gamma_1(c')    \right)\\
& \leq \sum_{c' \in [q]} \nu_v(c')\delta_{u,a,b}\FC\\
& = \EE{c' \sim \nu_v}{\delta_{u, a, b}\FC}\\
&\leq \max_{ a, b, c\in [q]} \EE{c' \sim \nu_v}{\delta_{u, a, b}\F}.
\end{align*}
This proves the inequality~\eqref{eq-opt-bound-1}.
\end{proof}

\begin{proof}(Proof of Claim~\ref{claim-optsol})
Suppose $\SOL^* = (\boldsymbol{S}_1^* , \boldsymbol{S}_2^* , c^*)$ is an optimal solution with $|S_2^*(u)| > 1$. 
Note that $S_2^*(u) \subset S_1^*(u)$.
Let $b \in [q]$ be an arbitrary element in $S_2^*(u)$. We remove the element $b$ from both $S_1^*(u)$ and $S_2^*(u)$ to obtain a new solution $\SOL^\circ = \{ \boldsymbol{S}_1^\circ, \boldsymbol{S}_2^\circ, c^\circ\}$. Namely
\begin{align*}
	S_1^\circ(u) &= S_1^*(u) \setminus \{b\}\\
	S_2^\circ(u) &= S_2^*(u) \setminus \{b\}
\end{align*}
and $c^\circ = c^*$, $S_1^\circ(w) = S_1^*(w)$, $S_2^\circ(w) = S_2^*(w)$ for all $w \in N_v \setminus \{u\}$. It is easy to verify that the new solution $\SOL^\circ$ also satisfies all the constraints. We will prove that $\SOL^\circ$ is also an optimal solution. Since $|S^\circ_2(u)| = |S^*_2(u)| - 1$, then we can repeat this argument to find the optimal solution $\SOL^\star$ with $|S^\star_2(u)| = 1$.

We denote the objective function value of the solution $\SOL^*$ as
\begin{align*}
g(\SOL^*) = \sum_{c' \in [q]}\nu_v(c')\left(\min_{\tau \in \mathcal{C}^*_2}\FA(\tau) - \min_{\tau \in \mathcal{C}^*_1}\FA(\tau)\right),
\end{align*}
where
\begin{align*}
\mathcal{C}_1^* = \bigotimes_{w \in N_v}S^*_1(w),\quad \mathcal{C}_2^* = \bigotimes_{w \in N_v}S^*_2(w).	
\end{align*}
Similar, we denote objective function value of the solution $\SOL^\circ$ as $g(\SOL^\circ)$. Suppose $S^*_1(u) \setminus S^*_2(u) = \{a\}$. We define the following set of values $S_a \subseteq [q]$ as
\begin{align*}
S_a \triangleq \left\{c' \in [q] \mid \min_{\tau \in \mathcal{C}^*_2}\FA(\tau) >  \min_{\tau \in \mathcal{C}^*_1}\FA(\tau) \right\}.
\end{align*}
Note that $\mathcal{C}_2^* \subset \mathcal{C}_1^*$. We must have $\min_{\tau \in \mathcal{C}^*_2}\FA(\tau) \geq   \min_{\tau \in \mathcal{C}^*_1}\FA(\tau) $ for all $c' \in [q]$.
Then $g(\SOL^*)$ can be rewritten as 
\begin{align}
\label{eq-g-SOL*}
g(\SOL^*) = \sum_{c' \in S_a}\nu_v(c')\left(\min_{\tau \in \mathcal{C}^*_2}\FA(\tau) - \min_{\tau \in \mathcal{C}^*_1}\FA(\tau)\right).	
\end{align}

For each $c' \in S_a$, suppose $\min_{\tau \in \mathcal{C}^*_1}\FA(\tau) = \FA(\tau^*)$, then it must hold that $\tau^*_u = a$. This is because $S^*_1(u)$ and $S^*_2(u)$ differ only at element $a$ and $S^*_1(w) = S^*_2(w)$ for all $w \in N_v \setminus \{u\}$. If $\tau^*_u \neq a$, we must have $\min_{\tau \in \mathcal{C}^*_2}\FA(\tau) = \min_{\tau \in \mathcal{C}^*_1}\FA(\tau)$. This is contradictory to $c' \in S_a$. 

Consider the solution $\SOL^\circ$. Note that the set $S^\circ_1(u) = S^*_1(u) \setminus \{b\}$ and  $S^\circ_2(u) = S^*_2(u) \setminus \{b\}$, where $b \neq a$ because $b \in S^*_2(u)$. By the definition of $\SOL^\circ(u)$, we have
\begin{align}
\forall c' \in S_a: \quad &\min_{\tau \in \mathcal{C}^*_1}\FA(\tau) = \min_{\tau \in \mathcal{C}^\circ_1}\FB(\tau)\label{eq-C1}\\
\forall c' \in S_a: \quad &\min_{\tau \in \mathcal{C}^*_2}\FA(\tau) \leq \min_{\tau \in \mathcal{C}^\circ_2}\FB(\tau),\label{eq-C2}
\end{align}
where 
\begin{align*}
\mathcal{C}_1^\circ = \bigotimes_{w \in N_v}S^\circ_1(w),\quad \mathcal{C}_2^\circ = \bigotimes_{w \in N_v}S^\circ_2(w).		
\end{align*}
Recall that  $c^\circ = c^*$, $S_1^\circ(w) = S_1^*(w)$, $S_2^\circ(w) = S_2^*(w)$ for all $w \in N_v \setminus \{u\}$. Thus ~\eqref{eq-C1} holds because $a \in S_1^\circ(u)$ and~\eqref{eq-C2} holds because $S^\circ_2(u) \subset S_2^*(u)$ (hence $\mathcal{C}^\circ_2 \subset \mathcal{C}^*_2$). Combining~\eqref{eq-C1} and~\eqref{eq-C2} together, we have 
\begin{align*}
g(\SOL^\circ)&=  \sum_{c' \in [q]} \nu_v(c')\left(\min_{\tau \in \mathcal{C}^\circ_2}\FB(\tau) - \min_{\tau \in \mathcal{C}^\circ_1}\FB(\tau)\right)\\
&\geq \sum_{c' \in S_a} \nu_v(c')\left(\min_{\tau \in \mathcal{C}^\circ_2}\FB(\tau) - \min_{\tau \in \mathcal{C}^\circ_1}\FB(\tau)\right)\\
&\geq \sum_{c' \in S_a}\nu_v(c')\left(\min_{\tau \in \mathcal{C}^*_2}\FA(\tau) - \min_{\tau \in \mathcal{C}^*_1}\FA(\tau)\right)\\
&= g(\SOL^*),
\end{align*}
where the last equation holds due to~\eqref{eq-g-SOL*}.
Thus $\SOL^\circ$ is also an optimal solution.	
\end{proof}